\documentclass[11pt]{article} %
\usepackage{times}
\usepackage[letterpaper, left=1in, right=1in, top=1in, bottom=1in]{geometry}

\usepackage[utf8]{inputenc} %
\usepackage[T1]{fontenc}    %

\usepackage{natbib}
\bibpunct{(}{)}{;}{a}{,}{,}

\usepackage[utf8]{inputenc} %
\usepackage[T1]{fontenc}    %
\usepackage{hyperref}       %
\usepackage{url}            %
\usepackage{booktabs}       %
\usepackage{amsfonts}       %
\usepackage{nicefrac}       %
\usepackage{microtype}      %
\usepackage{xcolor}         %

\usepackage{graphicx}
\usepackage{subcaption}
\usepackage{amsmath,amssymb, amsthm}
\usepackage[capitalize, noabbrev]{cleveref}
\usepackage{multirow}
\usepackage{multicol}
\usepackage{mathtools}
\usepackage{ marvosym }
\usepackage{setspace}
\usepackage{textcomp}
\usepackage{authblk} 
\usepackage{bbm}
\usepackage{array}
\usepackage{algorithm}
\usepackage[noend]{algpseudocode}
\usepackage{rotating}
\usepackage[shortlabels]{enumitem}
\usepackage{tikz}
\usetikzlibrary{bayesnet}

\usepackage{apptools}
\usepackage[page, header]{appendix}
\usepackage{titletoc}

\theoremstyle{plain}
\newtheorem{theorem}{Theorem}
\newtheorem{proposition}{Proposition}
\newtheorem{lemma}{Lemma}
\newtheorem{corollary}{Corollary}
\theoremstyle{definition}

\newtheorem{assumption}{Assumption}
\theoremstyle{remark}

\let\oldReturn\Return
\renewcommand{\Return}{\State\oldReturn}

\newcommand{\R}{\mathbb{R}}

\newcommand{\pr}{\textnormal{Pr}}
\newcommand{\indep}{\perp\!\!\!\!\perp} 
\DeclareMathOperator*{\ind}{1{\hskip -2.5 pt}\hbox{I}}  %

\newcommand{\Ycal}{\mathcal{Y}}
\newcommand{\E}{\mathbb{E}}
\newcommand{\Z}{\mathbb{Z}}

\newcommand{\Ncal}{\mathcal{N}}
\newcommand{\Dcal}{\mathcal{D}}
\newcommand{\opt}{\operatorname{OPT}}

\newcommand{\argmin}{\operatorname{argmin}}

\usepackage[textsize=tiny]{todonotes}

\hypersetup{
  colorlinks   = true, %
  urlcolor     = blue, %
  linkcolor    = blue, %
  citecolor   = red %
}

\usepackage{authblk}

\title{\textbf{Treatment response as a latent variable}}
\author[1]{Christopher Tosh}
\author[2]{Boyuan Zhang}
\author[1]{Wesley Tansey}
\affil[1]{Memorial Sloan Kettering Cancer Center, New York, NY}
\affil[2]{Stanford University, Palo Alto, CA}

\begin{document}
\maketitle

{\def\thefootnote{}
\footnotetext{E-mail:
\texttt{christopher.j.tosh@gmail.com},\
\texttt{boyuanz@stanford.edu},\
\texttt{tanseyw@mskcc.org}}}

\vspace{-2em}
\begin{abstract}

Scientists often need to analyze the samples in a study that responded to treatment in order to refine their hypotheses and find potential causal drivers of response. 
Natural variation in outcomes makes teasing apart responders from non-responders a statistical inference problem. 
To handle latent responses, we introduce the causal two-groups (C2G) model, a causal extension of the classical two-groups model.
The C2G model posits that treated samples may or may not experience an effect, according to some prior probability. 
We propose two empirical Bayes procedures for the causal two-groups model, one under semi-parametric conditions and another under fully nonparametric conditions. 
The semi-parametric model assumes additive treatment effects and is identifiable from observed data. 
The nonparametric model is unidentifiable, but we show it can still be used to test for response in each treated sample. 
We show empirically and theoretically that both methods for selecting responders control the false discovery rate at the target level with near-optimal power. 
We also propose two novel estimands of interest and provide a strategy for deriving estimand intervals in the unidentifiable nonparametric model. 
On a cancer immunotherapy dataset, the nonparametric C2G model recovers clinically-validated predictive biomarkers of both positive and negative outcomes.
Code is available at \url{https://github.com/tansey-lab/causal2groups}.
\end{abstract}

\section{Introduction}
\label{sec:introduction}

This article concerns causal inference when the treatment only affects a subset of the population. 
Such cases are common in medicine, where only individuals with certain characteristics experience the benefits of a prescribed intervention.
For example, a subpopulation of Hispanic patients has a germline mutation that prevents an approved CAR-T cell therapy from binding to the cell surface~\citep{seipel:etal:2023:snp-cd19-car-t-escape}. Lung tumors often acquire a somatic mutation that prevents certain small molecule inhibitors from binding~\citep{da:etal:2011:egfr-lung-cancer}. Personalized cancer mRNA vaccines are only effective if the neoantigens targeted provoke an immune response~\citep{rojas:etal:2023:personalized-mrna-pdac}. In all of these cases, biomarkers of response were unknown at the start of the study.
The key analysis step in refining the treatment criteria was to determine which patients likely experienced an affect on their outcomes. Since all patients experience natural variability in their outcomes, determining the affected patient set and the extent of the effects is a statistical inference problem.

In a dataset with treatment and control groups, what effects can be inferred if only some of the treated subjects receive any effect at all?
Can we identify treated individuals that saw an effect from the treatment?
Can we estimate the magnitudes of effects? What kinds of statistical guarantees can we provide?

To address these questions, we develop a general empirical Bayes framework built on top of the classic two-groups model \citep{efron:2008}. \cref{fig:graphical_models}d shows the causal two-groups (C2G) graphical model we propose here. We assume a binary treatment variable $T$, a binary latent response or effect variable $H$, continuous outcome $Y$, and observed confounders $X$. We also consider extensions that contain unmeasured confounders $U$ (see \cref{fig:confounding}). The central statistical challenge in the C2G setup is handling the confounding between the treatment ($T$), latent response ($H$), and outcome ($Y$) when conducting inference for each subject in the study.

The questions one can answer in the C2G model depend on the modeling assumptions one is willing to make. In common parametric and semi-parametric cases, the causal two-groups model can be shown to be identifiable.
In the fully nonparametric case, the C2G model and certain estimands of interest are unidentifiable.
Despite being unable to identify the nonparametric model, it turns out that it is possible to (i) develop an empirical Bayes procedure to test for individual causal effects, (ii) provide informative intervals for estimands of interest, and (iii) maintain robustness to certain kinds of latent confounding.

We propose two estimation algorithms for different instantiations of the C2G model. In the identifiable C2G regime, we develop an EM-based optimization routine for a semi-parametric model under an additive errors assumption. Given a fitted model, there is a step-down empirical Bayes approach \citep{efron:2004, efron2012} to select individuals that came from the responder distribution. We show that this approach controls the false discovery rate (FDR) and attains near-optimal power (\cref{thm:additive-fdr}).

Our second approach handles the fully nonparametric setting. Even though there are many possible prior/responder distribution pairs that could give rise to the observed treatment distribution, it can be shown that this set only enjoys a single degree of freedom, corresponding to a closed interval (\cref{lem:conservative-pi}). Moreover, the endpoints of this interval can be calculated in closed form given estimates of the observed untreated and treated distributions. Using this result, we can construct a valid conservative selection procedure that controls FDR and attains near-optimal power (\cref{thm:general-fdr}). Moreover, we can also easily derive intervals that provably contain the ground-truth latent effect (\cref{corr:np-ite}).

Beyond the latent drug response scenario, the C2G model has applications in many other areas of medicine and science. In particular, the C2G model generalizes the intention-to-treat setting, where a treatment is assigned but compliance with treatment is unmeasured. In the out-patient setting, a patient is prescribed a medication, sent home, and may or may not actually take the treatment. The only observation available to the healthcare provider is when the patient is readmitted for subsequent treatment, leaving the provider unsure whether the patient truly complied with the original treatment assignment. Similarly, social welfare programs like housing lotteries may not actively track which specific vouchers are redeemed, jobs training programs may not take attendance, and
food subsidy programs are unable to know whether foods purchased were actually consumed by the purchaser~\citep{hidrobo:etal:2014:cash-food-vouchers}. In all of these instances, the C2G model can be reinterpreted to have $T$ be the intention-to-treat and $H$ be the latent treatment compliance status. This paper shows the limits and capabilities of causal inference in such studies, and provides new tools to make causal conclusions in the face of latent treatment effects. %

\section{Background}
\label{sec:background}
This paper draws inspiration from three distinct areas in statistics and machine learning. The first of these is multiple hypothesis testing, where there are many separate hypotheses to be tested and the primary objective is to maximize power while controlling FDR. The second area is treatment assignment and noncompliance, a rich field that has been explored extensively in the causal inference literature. And the last area is the recent body of work in machine learning on estimating treatment effects using flexible black box models. Below, we survey related work in each of these and contrast it with the causal two-groups setting considered in this paper.

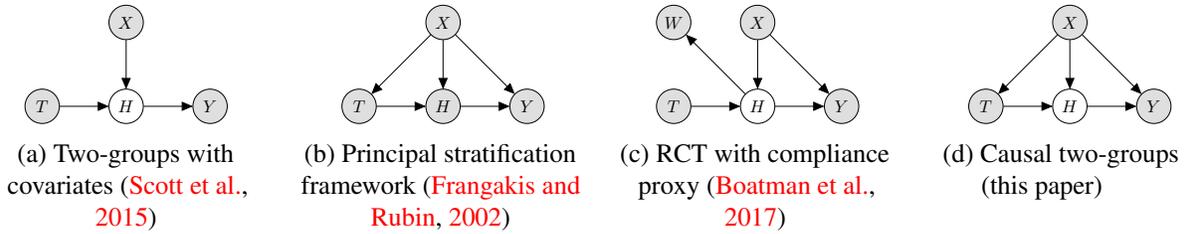
\begin{figure*}[t]
\centering
\begin{subfigure}{0.23\linewidth}
\captionsetup{justification=centering}
\centering
\scalebox{0.65}{\begin{tikzpicture}

  \node[obs, ]                               (x) {$X$};
  \node[latent, below=of x]                (h) {$H$};
  \node[obs, left=of h]                   (t) {$T$};
  \node[obs, right=of h]                   (y) {$Y$};

  \edge {x} {h}  ; %
  \edge {t} {h}  ; %
  \edge {h} {y}  ; %

\end{tikzpicture}}
\caption{\label{fig:2g_model}Two-groups with covariates~\citep{scott:etal:2014:fdr-regression}}\end{subfigure}\quad
\begin{subfigure}{0.23\linewidth}
\captionsetup{justification=centering}
\centering
\scalebox{0.65}{
\begin{tikzpicture}

  \node[obs, ]                               (x) {$X$};
  \node[obs, below=of x]                (h) {$H$};
  \node[obs, left=of h]                   (t) {$T$};
  \node[obs, right=of h]                   (y) {$Y$};

  \edge {x} {t} ; %
  \edge {x} {h}  ; %
  \edge {x} {y}  ; %
  \edge {t} {h}  ; %
  \edge {h} {y}  ; %

\end{tikzpicture}}
\caption{\label{fig:rct_model}Principal stratification framework~\citep{frangakis:rubin:2002:principal-stratification} }\end{subfigure}\quad
\begin{subfigure}{0.23\linewidth}
\captionsetup{justification=centering}
\centering
\scalebox{0.65}{
\begin{tikzpicture}

  \node[obs, ]                               (x) {$X$};
  \node[latent, below=of x]                (h) {$H$};
  \node[obs, left=of x]                     (w) {$W$};
  \node[obs, left=of h]                   (t) {$T$};
  \node[obs, right=of h]                   (y) {$Y$};

  \edge {x} {h}  ; %
  \edge {x} {y}  ; %
  \edge {t} {h}  ; %
  \edge {h} {y}  ; %
  \edge {h} {w}  ; %

\end{tikzpicture}}
\caption{\label{fig:rct_proxy}RCT with compliance proxy~\citep{boatman:etal:2017:compliance-error}}\end{subfigure}\quad
\begin{subfigure}{0.23\linewidth}
\captionsetup{justification=centering}
\centering
\scalebox{0.65}{\begin{tikzpicture}

  \node[obs, ]                               (x) {$X$};
  \node[latent, below=of x]                (h) {$H$};
  \node[obs, left=of h]                   (t) {$T$};
  \node[obs, right=of h]                   (y) {$Y$};

  \edge {x} {t} ; %
  \edge {x} {h}  ; %
  \edge {x} {y}  ; %
  \edge {t} {h}  ; %
  \edge {h} {y}  ; %

\end{tikzpicture}}
\caption{\label{fig:causal_2g_model}Causal two-groups \newline (this paper) \newline \mbox { }}
\end{subfigure}
\caption{\label{fig:graphical_models}
(a) The FDR regression model of \citet{scott:etal:2014:fdr-regression}.
(b) An observational dataset with observed noncompliance \citep{frangakis:rubin:2002:principal-stratification}.
(c) A randomized controlled trial with a compliance proxy \citep{boatman:etal:2017:compliance-error}.
(d) The general causal two-groups model. The treatment $T$, (unobserved) treatment effect indicator $H$, and the outcome $Y$ are all confounded by $X$. The causal two-groups model generalizes other scenarios such as those in (a) and (b).}
\end{figure*}

\subsection{Multiple hypothesis testing with covariates}
\label{subsec:background:multiple-testing}
Several approaches to multiple testing incorporate covariates as side information. These methods generally work under the two-groups model \citep{efron:2004,efron:2008,efron2012} where a test statistic is drawn either from the null or alternative with some prior probability. However, both frequentist \citep{ignatiadis:etal:2016:ihw,xia:etal:2017:neuralfdr,lei:fithian:2018:adapt,chen:etal:2018:functional-fdr,li:barber:2019:sabha} and Bayesian \citep{scott:etal:2014:fdr-regression,patra:bodhi:2016:2groups-jrssb,tansey:etal:2018:fdr-smoothing,tansey:etal:icml:2018:bbfdr} approaches assume that the covariates only impact the likelihood of treatment having an effect. This is often a reasonable assumption in laboratory experiments, where biological replicates can be tested repeatedly under different experimental conditions. In observational data, the assumption is unlikely to hold as covariates are often confounders that affect the probability of treatment assignment, treatment compliance, and both the null and alternative outcomes. \Cref{fig:2g_model} shows the graphical model for existing two-groups models.

In a similar vein to the latent response model considered here, \citet{duan:etal:2024:interactive-fdr-control} propose an interactive protocol to the testing problem, in which an analyst is given the covariates and outcomes but is blinded to the treatment assignments, and they must cooperate with an algorithm that knows the treatment assignments to make discoveries. In their setup, however, the blinding to treatment assignments is not intrinsic to the problem but is rather a technique to increase power without losing FDR control.

\subsection{Causal inference with noncompliance}
\label{subsec:background:compliance}
In the statistics literature, methods for handling noncompliance were originally motivated by clinical trial data. \citet{efron:feldman:1991:compliance} observed that patients receiving placebo treatments had better outcomes as a function of compliance rate to the placebo. \citet{frangakis:rubin:2002:principal-stratification} generalized this in the \textit{principal stratification} framework, with treatment assignment and compliance jointly representing potential interventions. \Cref{fig:rct_model} shows the graphical model for a randomized controlled trial (RCT) in the principal stratification setting. 
These models assume compliance is observed, leading to identifiability of the causal effects; here we consider the case where compliance is entirely latent.

Other methods consider noisy or latent compliance with side information. \citet{angrist:etal:1996:iv-noncompliance} introduced an instrumental variable approach to estimating average treatment effects with noisy compliance. \citet{boatman:etal:2017:compliance-error} consider the scenario where a noisy proxy variable is available for compliance, such as blood testing for nicotine levels when patients may be misreporting their smoking habits. Further, they assume the true parametric model is known and identifiable. \Cref{fig:rct_proxy} shows the graphical model for the proxy variable case. More recent work leverages auxiliary variables combined with conditional ignorability assumptions or partial observability of the compliance variable~\citep{jiang:ding:2021:identification-principal-strata-auxiliary,jiang:etal:2022:multiply-robust-principal-ignorability,lu:etal:2023:principal-strata-continuous}.  These models obtain identifiability by leveraging additional information to disentangle the confounding from the causal effects. By contrast, we consider the challenging case where no instrumental or proxy variables are available.

\subsection{Estimating individual treatment effects in observational data}
\label{subsec:background:causal-effects}
A large body of recent work in machine learning has considered the problem of predicting the conditional average treatment effect~\cite[e.g.][]{shi:etal:2019:dragonnet,johansson:etal:gen-bounds-causal,bica:etal:2020:causal-gan,wang:etal:2024:cauusal-optimal-transport}. In the presence of latent response or noncompliance, the conditional average treatment effect (CATE) is generally an underestimate of the true effect if one could intervene, such as requiring compliance. CATE estimates also do not directly help identify which samples may have responded to treatment, as they integrate out the probability of response in a given sample. However, methods that provide uncertainty quantification can be used to conduct hypothesis tests for whether a given treatment has any effect on a given individual. The causal forests method \citep{wager:athey:2018:causal-forests} adapts random forests \citep{breiman:2001:random-forests} to the causal setting and provides asymptotic confidence intervals for effects. Bayesian additive regression trees (BART) \citep{chipman:etal:2010:bart} have been shown to perform well for estimating causal effects without any modification \citep{hill:etal:2011:causal-bart}. The posterior credible intervals provided by BART provide the Bayesian equivalent of frequentist confident intervals for causal effect sizes. Both causal forests and BART are predictive models with no explicit notion of null effects or noncompliance. In principle, the predictive models could be integrated into our causal two-groups procedures in place of black box neural networks. Synthetic and semi-synthetic simulations in \cref{sec:more_simulations,sec:gdsc_simulations} show that confidence intervals from these models are not sufficient to control FDR in finite samples.

\section{The causal two-groups model}
\label{sec:model}
Let $(x_i, y_i, t_i)_{i=1}^n$ be an observational dataset of $n$ samples, where $t_i \in \{0,1\}$ indicates whether individual $i$ was treated $(t_i=1)$ or not $(t_i=0)$, $y_i \in \R$ is the individual outcome, and $x_i \in \R^d$ are potential confounders. For instance, $t_i=1$ indicates that the patient was prescribed a certain drug, $y_i$ is the patient's survival time in days, and $x_i$ is demographic information like age, race, and sex. We drop the subscripts from here on when it is clear we are referring to a single observation.

We focus on the primary analysis goal of selecting responders $(h = 1)$; we will consider estimation of effect sizes as a secondary goal. We adapt the two-groups (2G) model from the multiple testing literature \citep{efron:2004,efron:2008}. For each treated individual ($t=1$), we suppose they are either \textit{non-responders} or \textit{responders}, with outcome distributions $f_0(y \mid x)$ or $f_1(y \mid x)$, respectively. All untreated ($t=0$) patients are assumed to have the same outcome distribution as non-responders. Each treated individual is affected by the treatment with probability $\pi(x)$,
\begin{equation}
\label{eqn:two_groups}
\begin{aligned}
(y \mid x, h) &\sim& (1-h) f_0(y \mid x) + h f_1(y \mid x) \\
(h \mid x, t=0) &=& 0 \\
(h \mid x, t=1) &\sim& \mbox{Bernoulli}(\pi(x)) \\
(t \mid x) &\sim& \mbox{Bernoulli}(\phi(x))  \, ,
\end{aligned}
\end{equation}
where $h \in \{0,1\}$ indicates whether the individual saw any effect from the treatment. Equivalently, $h$ can capture whether the individual complied with a treatment in the intent-to-treat scenario; we will use the responder terminology for the remainder of the paper. \Cref{fig:causal_2g_model} shows the causal graphical model corresponding to \cref{eqn:two_groups}. The model assumes $X$ is a confounder of $T$, $H$, and $Y$. Scenarios with no treatment bias (e.g. randomized controlled trials) are covered as well, as \cref{eqn:two_groups} is strictly a generalization. The model also assumes that treatment response is random, allowing for exogenous variables (e.g. unmeasured biomarkers) to induce uncertainty in response.

\subsection{When is the causal two-groups model identifiable?}

Our first result is that, with no further assumptions, the causal two-groups model is not identifiable, even without overlap violations.
\begin{proposition}
\label{thm:unident}
The model in \cref{eqn:two_groups} is unidentifiable without further assumptions, even when restricting $0 < \pi(x) < 1$ for all $x$. Moreover, it is unidentifiable even when $x$ is a constant and the distributions $f_0, f_1$ are restricted to being differentiable log-concave probability distributions.
\end{proposition}
For space and clarity, all proofs are presented in the appendix. The key idea is to rewrite the treatment outcome distribution $f_t(y \mid x) := p(y \mid x, t=1)$ as a mixture model,
\begin{equation}
\label{eqn:treatment-mixture}
f_t(y \mid x) = (1-\pi(x)) f_0(y \mid x) + \pi(x) f_1(y \mid x).
\end{equation}
Then one can show that even when $f_0(y \mid x)$ is fixed, there is a range of $\pi$'s and $f_1$'s that give rise to exactly the same $f_t$.

The story changes if parametric forms of the conditional likelihoods $f_0, f_1$ are known. Indeed, the form of \Cref{eqn:treatment-mixture} implies that \Cref{eqn:two_groups} is identifiable at those $x$ such that $f_0(\cdot \mid x), f_1(\cdot \mid x)$ come from a family for which finite mixtures are identifiable~\citep{teicher:1963:identifiability-finite, yakowitz:spragins:1968:identifiability-finite}. As a concrete example, we give a direct proof that \Cref{eqn:two_groups} is identifiable for the class of conditional normal likelihoods.

\begin{proposition}
\label{prop:normal-identifiability}
Suppose that $f_i$ is of the form
$f_i(y \mid x) = \Ncal(y \mid \mu_i(x), \sigma_i^2(x))$
for $i=0,1$. Then \Cref{eqn:two_groups} is identifiable for all $x$ such that $f_0(\cdot \mid x) \neq f_1(\cdot \mid x)$.
\end{proposition}

In practice, we generally do not know the parametric forms of the outcome distributions a priori. Thus, we would ideally like to avoid parametric assumptions as much as possible.
As we show in \cref{sec:additive,sec:kernel_generalized}, there is some hope here: nonparametric estimation of the true model may be impossible, but semi-parametric estimation and nonparametric testing are viable. The upshot here is that we can select responders with minimal modeling assumptions while still controlling the desired error rate.

\subsection{Testing in the causal two-groups model}
\label{subsec:model:empirical_bayes_testing}
We formulate responder selection as a multiple hypothesis testing problem.
For each observed outcome in a treated sample, we are testing the null hypothesis that it was drawn from the non-responder distribution,
\[
H_0 \colon h = 0 \, .
\]
The multiple testing goal is to maximize power while controlling the false discovery rate at or below the target $\alpha$ level. The false discovery rate selecting $\hat{V} \subseteq \{ i \colon 1 \leq i \leq n \, , t_i = 1 \}$ is given by
\[ \text{FDR} = \E\left[ \frac{1}{|\hat{V}|} \sum_{i \in \hat{V}} \ind[h_i = 0] \right] .\]
As in other two-groups models~\citep{efron:2004,efron:2008,efron2012,scott:etal:2014:fdr-regression,tansey:etal:2018:fdr-smoothing,tansey:etal:icml:2018:bbfdr}, we take an empirical Bayes approach.
At a high level, empirical Bayes testing in the causal two-groups model proceeds as follows:
\begin{enumerate}
    \item Fit $\hat{f}_0(y \mid x)$, a model of outcomes for the untreated population. This also serves as the outcome model for non-responders by definition in \cref{eqn:two_groups}.
    \item Provide (conservative) estimates $\hat{f}_1(y \mid x)$ and $\hat{\pi}(x)$ of $f_1(y \mid x)$ and $\pi(x)$, the heterogeneous treatment effect model and the prior probability of seeing an effect, respectively, using the outcomes of the treated population.
    \item Using the estimates $\hat{f}_0$, $\hat{f}_1$, and $\hat{\pi}$, calculate the (conservative) posterior probability of each treated sample having come from the non-responder distribution,
    \[ \hat{w}_i = \frac{(1-\hat{\pi}(x_i))\hat{f}_0(y_i \mid x_i)}{(1-\hat{\pi}(x_i))\hat{f}_0(y_i \mid x_i) + \hat{\pi}(x_i)\hat{f}_1(y_i \mid x_i)}. \]
    \item Select discoveries $\hat{V} \subseteq \{ i \colon 1 \leq i \leq n \, , t_i = 1 \}$ at the $\alpha$ FDR level using the estimates $\hat{w}_i$,
            \begin{equation}
            \label{eqn:selection}
            \begin{aligned}
            \underset{V}{\text{maximize}} & \mbox{ } |V| \\
             \text{subject to} & \mbox{ } \frac{1}{|V|} \sum_{i \in V} \hat{w}_i \leq \alpha \, .
            \end{aligned}
            \end{equation}
\end{enumerate}
The details of how to model $f_0$, $f_1$, and $\pi$ are dependent on the assumptions one is willing to make about the true causal model. We will consider two such models, one simpler and the other more general.
\Cref{sec:additive} presents the simpler model that assumes additive noise while keeping the noise distribution and outcome mean functions nonparametric.
\Cref{sec:kernel_generalized} then considers arbitrary distributions for the outcomes and response probabilities. 

\subsection{Estimands of interest in the causal two-groups model}
\label{subsec:model:estimands}
Latent (non-)responders in the treatment group make interpreting causal effects more nuanced than in the traditional causal inference setup. We will consider two causal estimands that may be of interest to the analyst.

\begin{itemize}
    \item \emph{Conditional average response effect (CARE)}. This is the analogue of the classical conditional average treatment effect (CATE). Statistically,
    $$\E[Y \mid t=1, h=1, X=x] - \E[Y \mid t=0, X=x] \, .$$
    The CARE estimates the impact of the treatment if we could intervene and force response on treated samples. For instance, rather than prescribing a drug in the out-patient setting, we could instead administer the treatment in the clinic. Removing conditioning yields the average response effect (ARE), analogous to the average treatment effect (ATE). For positive effects, the (C)ATE will generally yield a biased downward estimate of the true effect captured by the (C)ARE.

    \item \emph{Expected responder population fraction (ERPF)}. This answers the question of how many people did a treatment effect. Statistically,
    $$\frac{1}{\sum_i t_i} \sum_i t_i \hat\pi(x_i) \, .$$
    In the classical setup, the ERPF is simply one. When only a fraction of patients respond, the ERPF enables us to distinguish between treatments with a strong effect for a few patients and treatments with a small but consistent effect across the majority of a population.
\end{itemize}

Our ability to estimate these quantities depends on whether the modeling assumptions provided render the corresponding estimand identifiable or at least allow for interval estimation. When only an interval is identifiable, as we shall see in \cref{sec:kernel_generalized}, we can use the lower end to provide conservative estimates.

The potential for non-response also leads to risk-reward trade-offs. For a given set of covariates, we can estimate both the expected response fraction and the expected effect size of the response. Deciding whether to prescribe a treatment then involves a delicate balance in the face of potential treatment costs and side effects. \cref{sec:case-study} shows an example of estimating the CARE and ERPF on a population of cancer patients treated with combination immunotherapy.

\section{A semi-parametric additive errors model}
\label{sec:additive}
In the additive errors (AE) regime, we model $y$ as a deterministic function plus mean zero i.i.d. noise,
\begin{equation}
\label{eqn:additive_errors}
\begin{aligned}
(y \mid x, h=0) &=& \mu_0(x) + \epsilon  & &\\
(y \mid x, h=1) &=& \mu_1(x) + \epsilon &=& \mu_0(x) + \tau(x) + \epsilon\\
\epsilon &\sim& g(\epsilon)\,, && \mathbb{E}[\epsilon] = 0  \, , 
\end{aligned}
\end{equation}
where $\tau(x) = \mu_1(x) - \mu_0(x)$ is the expected difference in outcome between the responder and non-responder models. \cref{eqn:additive_errors} is common in many causal inference setups, particularly when estimating the conditional average treatment effect \citep[cf.][]{hahn:etal:2017:bcf,wager:athey:2018:causal-forests}. This setting can be thought of as a conditionally semi-parametric model. Conditioned on covariates $x$, the model consists of the parameters $\mu_0(x), \mu_1(x)$ and the infinite-dimensional model for the noise $\epsilon$. The AE model is strictly more flexible than a parametric model with a location parameter. Despite the added flexibility, we can show that this model is identifiable, meaning we can directly estimate the underlying parameters.

\begin{theorem}
\label{thm:additive-identify}
The model in \cref{eqn:additive_errors} is identifiable at every $x$.
\end{theorem}

The upshot to \cref{thm:additive-identify} is that we can fit a semi-parametric model to \cref{eqn:additive_errors} and use the resulting model to test the hypothesis $H_0: \mathbb{E}[y - \mu_0(x)] = 0$. 
We can also directly interpret $\tau(x)$ as the CARE for a sample with covariates $x$. 
Additionally, after fitting a semi-parametric model to \cref{eqn:additive_errors}, we can also recover the ERPF from our estimate of $\pi$.

To fit a model to \cref{eqn:additive_errors}, we follow a stagewise procedure. First, we fit a nonparametric regression function $\hat{\mu}_0(x)$ to the expected outcomes for the untreated population. Next, we use a nonparametric density estimator to marginally model the residual error distribution $\hat{g}$ of $y - \hat{\mu}_0(x)$ on the untreated population. Finally, we use an EM algorithm to fit nonparametric regression functions $\hat{\pi}$ and $\hat{\mu}_1$ for the prior and responder model, respectively, on the treated population, utilizing $\hat{g}$ and $\hat{\mu}_0$. The remainder of this section details the above approach.

\subsection{Estimating the non-response distribution}
\label{subsec:additive:null}
We model $\hat{\mu}_0$ using kernel ridge regression with a radial basis function kernel, and we fit it on the untreated population, using generalized cross-validation to select the bandwidth and regularization parameters~\citep{golub:etal:1979:gen-cv}. After fitting $\hat{\mu}_0$, we compute the leave-one-out predictions for each point in the untreated population, which can be calculated efficiently in closed form. We use these conditional expectations to estimate the residual noise distribution $g$. Cross-validation generally produces a biased-upward estimate of the true error distribution \citep{efron:gong:1983:leisurely-cv}. By overestimating the tails of the residual distribution, we conservatively bias the posterior probability of a response downward and the local FDR estimates upward.

We estimate the distribution of residuals $\hat{g}$ nonparametrically. We use predictive recursion \citep{martin:tokdar:2009:predictive-recursion} for its good rates of convergence and strong performance in other two-groups models \citep{scott:etal:2014:fdr-regression,tansey:etal:2018:fdr-smoothing}.
We choose the bandwidth by maximizing the predictive recursion marginal likelihood (PRML) \citep{martin:tokdar:2011:pr-marginal-likelihood}. This yields a tighter fit to the data and makes the density estimation routine fully auto-tuned and data-adaptive.

\subsection{Estimating the prior and response distributions}
\label{subsec:additive:alternative}
With $\hat{g}$ and $\hat{\mu}_0$ in hand, we form an estimate of the alternative distribution for the treated population. Specifically, we model each member of the treatment population as arising from the mixture model
\begin{equation}
\label{eqn:residual_mixture}
(y \mid x, t=1) \sim (1-\pi(x)) \hat{g}(y - \hat{\mu}_0(x)) + \pi(x) \hat{g}(y - \mu_1(x)) \, .
\end{equation}
We fit \cref{eqn:residual_mixture} via expectation-maximization with the following steps.

\paragraph{E-step.} Fix the estimates $\hat{\pi}$ and $\hat{\mu}_1$. Calculate the expected posterior probability of each data point coming from $\hat{g}(y - \mu_0(x))$,
\begin{equation*}
\label{eqn:additive_e_step}
\hat{w}_i = \frac{(1-\hat{\pi}(x_i)) \hat{g}(y_i - \hat{\mu}_0(x))}{(1-\hat{\pi}(x_i)) \hat{g}(y_i - \hat{\mu}_0(x)) + \hat{\pi}(x_i) \hat{g}(y_i - \hat{\mu}_1(x))} \, .
\end{equation*}
The posterior expectations $\hat{w}$ then serve as weights in the M-step.

\paragraph{M-step.} Fixing the weights $\hat{w}$, the optimization problem is separable in $\mu_1$ and $\pi$,
\begin{equation*}
\label{eqn:additive_m_step}
\begin{aligned}
\hat{\pi} &= \underset{\pi}{\text{argmax}} \sum_{i=1}^n \left[ \hat{w}_i \log(1-\pi(x_i)) + (1-\hat{w}_i) \log(\pi(x_i)) \right] \\
\hat{\mu}_1 &= \underset{\mu_1}{\text{argmax}} \sum_{i=1}^n (1-\hat{w}_i) \log(\hat{g}(y_i - \mu_1(x_i))) \, .
\end{aligned}
\end{equation*}
Both $\hat{\mu}_1$ and $\hat{\pi}$ (in logit-space) are encoded as linear models on top of random Fourier features, which approximate the full kernelized models and can be fit using gradient-based methods~\citep{rahimi:recht:2007:random-fourier-features}. The bandwidth parameters are selected using $k$-fold cross validation, and final predictions for the treated population are made on each held-out fold.

\subsection{Selecting responders}
\label{subsec:method:selection}
The final E-step estimate $\hat{w}$ has the appeal of both frequentist and Bayesian interpretations. Frequentists can interpret $\hat{w}\times 100$ as a local false discovery rate. Bayesians can interpret $1-\hat{w}$ as a posterior probability of treatment response. To select responders, we order the $\hat{w}$ values in ascending order and select the largest subset $\hat{S} \subset [n]$ such that $\frac{1}{|\hat{S}|} \sum_{i \in \hat{S}} \hat{w}_i \leq \alpha$, for a target $\alpha$-level FDR. While FDR control is guaranteed if we have the true posteriors~\citep{efron2012}, in practice with finite samples we will have some estimation error. Under some regularity assumptions, we can show that that excess FDR can be bounded above as a function of the estimation error.

\begin{assumption}
\label{assump:add-assumptions}
Let $\epsilon, \lambda, L, U > 0$ be given. Suppose the following holds for all $i=1,\ldots, n$:
\begin{itemize}
    \item[(i)] $g$ is $\lambda$-Lipschitz and bounded above by $U$,
    \item[(ii)] $(1-\pi(x_i))g(y_i - \mu_0(x_i)) + \pi(x_i) g(y_i - \mu_1(x_i)) \geq L$,
    \item [(iii)] $|\pi(x_i) - \hat{\pi}(x_i)|, |\mu_0(x_i) - \hat{\mu}_0(x_i)|, |\mu_1(x_i) - \hat{\mu}_1(x_i)| \leq \epsilon$, and
    \item[(iv)] $|g(y) - \hat{g}(y)| \leq \epsilon$ for all $y \in \R$.
\end{itemize}
\end{assumption}
Assumptions (i) and (ii) require the residual distribution and the treatment distributions to be well-behaved and are similar to standard assumptions made in the density estimation literature. Assumptions (iii) and (iv) reflect the quality of the estimators. This approximation error linearly translates to FDR.
\begin{theorem}
\label{thm:additive-fdr}
Suppose \cref{assump:add-assumptions} holds. If $\epsilon \leq \min \{1, L/4(U+2(\lambda + 1)) \}$, then the procedure outlined above results in FDR bounded by $\alpha + \delta$, where $\delta = \frac{2U(U+2(\lambda + 1))}{L^2} \epsilon$.
\end{theorem}
Thus, so long as we are reasonably accurate in our estimates, the realized FDR will be close to the target FDR.
Further, the power of the procedure can also be bounded from below as a function of the error, where the power of a procedure that selects a subset $V \subset \{ i \in [n]: t_i=1 \}$ is the expected fraction of individuals with $h_i = 1$, i.e. $\E\left[ \frac{1}{n} \sum_{i \in V} \ind[h_i = 1] \right]$. As in the FDR bound, approximation error also translates linearly to power.
\begin{theorem}
\label{thm:additive-power}
Suppose \cref{assump:add-assumptions} holds. The power of the procedure is bounded below by 
\[ \frac{1-\alpha - \delta}{1-\alpha + \delta + 1/(n_{\opt}(\alpha - \delta) +1)} \opt(\alpha - \delta), \]
where $\opt(\beta)$ is the power of the Bayes' optimal procedure that achieves FDR bounded above by $\beta$, $n_{\opt}(\beta)$ is the corresponding number of selections, and $\delta = \frac{2U(U+2(\lambda + 1))}{L^2} \epsilon$.
\end{theorem}
The number of selections plays a role in the quality of the power approximation in \cref{thm:additive-power}. This factor arises due to the discrete nature of the selection problem and the fact that it may not be possible to choose a subset with predicted FDR exactly equal to some quantity. However, as the number of optimal selections increases, the approximation ratio tends to $(1-\alpha-\delta)/(1-\alpha + \delta)$. Thus, with small model error, the selection procedure tends toward optimal power.

\section{Causal two-groups with nonparametric models}
\label{sec:kernel_generalized}
The additivity assumption in \cref{eqn:additive_errors} is often violated in practice. For instance, \citet{efron:feldman:1991:compliance} observed that higher rates of compliance were associated with increased variability in outcomes for both the treatment and control groups. Further, the increase in variance differed between treatment and control groups, suggesting that the error models for the two distributions $f_0$ and $f_1$ were different in the study.

Here, we propose a generalized empirical Bayes procedure for estimation in the causal two-groups model. The goal is to be maximally flexible to the true data generating model. As such, we remove the assumption of additive errors and instead directly model the outcome distributions in \cref{eqn:two_groups}. The starting point for this approach is the following result.
\begin{lemma}
\label{lem:conservative-pi}
Suppose the generative model from \cref{eqn:two_groups} induces treated distribution $f_t(y \mid x)$ and untreated distribution $f_0(y \mid x)$. Then for any $x$, we have
\[ \pi(x) \geq \pi^\star(x) := 1 - \min_y \frac{f_t(y \mid x)}{f_0(y \mid x)}. \]
Moreover, for any function $\pi$ satisfying $\pi(x) \in [\pi^\star(x), 1]$ for all $x$, there is an alternative distribution $f^\pi_1$ such that
\begin{align*}
f_t(y \mid x) &= (1-\pi(x)) f_0(y \mid x) + \pi(x) f^\pi_1(y \mid x).
\end{align*}
\end{lemma}
The implication of \cref{lem:conservative-pi} is that given the observed null and treatment distributions, one can compute the most conservative mixing function, $\pi^\star$, by minimizing their ratio. This leads to a valid test for an observation $(x,y)$ that controls the Type I error rate:
\begin{equation}
\label{eqn:general-bayes-test}
    \hat{h}_{\text{Bayes}} = \ind\left[ \frac{(1 -\pi^\star(x)) f_0(y \mid x) }{\pr(y \mid x, t=1)} \leq \alpha  \right] .
\end{equation}
Thus, empirical Bayesian testing in the fully general causal two-groups model can be reduced to finding suitable estimates of $f_0$ and $f_t$.

\subsection{Estimating the non-response and treatment distributions}
\label{subsec:kernel_generalized:estimation}
We take a nonparametric approach to approximating $f_0$ and $f_t$, using a $k$ nearest-neighbor version of the Rosenblatt kernel conditional density estimator~\citep{holmes:etal:2007:kernel-density-estimation,rosenblatt:1969:kernel-density-estimation}. That is, we approximate $f_t$ as
\[ \hat{f_t}(y \mid x ) = \frac{\sum_{i=1}^k K_{h_1}(x \mid x_{j_{x,1,i}})K_{h_2}(y \mid y_{j_{x,1,i}})}{\sum_{i=1}^k K_{h_1}(x \mid x_{j_{x,1,i}})}, \]
where $K_h$ is a probability kernel with bandwidth parameter $h>0$ and $j_{x,t,i}$ is the index of the $i$-th nearest neighbor of $x$ with treatment value $t$ in the dataset. $\hat{f_0}$ is constructed identically, substituting $j_{x,0,i}$ for $j_{x,1,i}$. The bandwidth parameters $h_1, h_2$ and number of neighbors $k$ are selected via leave-one-out cross validation.

\subsection{Selecting responders}
\label{subsec:kernel_generalized:estimation}
The test in \cref{eqn:general-bayes-test} relies on ratios of densities, making it sensitive to errors in our estimates. To address this issue, we use bootstrap sampling to create a population of conditional densities and estimate the ratios by using upper/lower quantiles to produce conservative ratio estimates. The final null posterior probability estimate for observation $x_i, y_i$ is given by
\begin{align}
\hat{w}_i = \frac{\hat{f}_{0,1-q}(y_i \mid x_i) }{\hat{f}_{t,q}(y_i \mid x_i)} \cdot \min_{y} \frac{\hat{f}_{t,1-q}(y \mid x_i)}{\hat{f}_{0,q}(y \mid x_i)},
\label{eqn:generalized_w_estimate}
\end{align}
where $\hat{f}_{t,q}(y\mid x)$ is the $q$-th quantile (under the bootstrap distribution) of the conditional treatment density of $y$ given $x$. Similarly, $\hat{f}_{0,q}(y\mid x)$ is the $q$-th quantile of the conditional untreated distribution. As in the additive case, we can sort the $\hat{w}$ values and select the largest subset such that their average value is below $\alpha$. The following result shows that when the conditional distributions are estimated accurately enough, the procedure results in proper FDR control.

\begin{theorem}
\label{thm:general-fdr}
Let $\epsilon, L, U > 0$ be given. Suppose the following holds for all $i=1,\ldots, n$ and $y \in \R$:
\begin{itemize}
    \item[(i)] $L \leq f_t(y \mid x_i), f_0(y \mid x_i) \leq U$,
    \item [(ii)] $|f_t(y \mid x_i) - \hat{f}_t(y \mid x_i)| \leq \epsilon$, and
    \item[(iii)] $|f_0(y \mid x_i) - \hat{f}_0(y \mid x_i)| \leq \epsilon$.
\end{itemize}
If $\epsilon \leq \min(1, L/2)$, then the procedure outlined above results in FDR bounded by $\alpha + \delta$, where $\delta = \frac{8U}{L^2} \left( 1 + \frac{12U}{L^2}\right)\epsilon$. 
Moreover, the power of the procedure is bounded below by 
\[ \frac{1-\alpha - \delta}{1-\alpha + \delta + 1/(n_{\opt}(\alpha - \delta) +1)} \opt(\alpha - \delta), \]
where $\opt(\beta)$ is the power of the Bayes' optimal procedure under the conservative prior of \Cref{lem:conservative-pi} that achieves FDR bounded above by $\beta$, and $n_{\opt}(\beta)$ is the corresponding number of selections.
\end{theorem}

We further control FDR with an empirical estimate of the FDR, calculated by performing the above selection procedure on both the treated and untreated populations over a grid of nominal FDR values $\alpha_1 < \alpha_2 < \cdots < \alpha_m$, and taking $e_k$ to be the ratio of the number of untreated selected at level $\alpha_k$ over the number of treated selected at level $\alpha_k$. Selection at level $\alpha$ proceeds by finding the largest $\alpha_k \leq \alpha$ such that $e_k \leq \alpha_k$ and selecting at level $\alpha_k$. We call this procedure \emph{empirical control}.

\subsection{Estimand intervals}
\label{subsec:kernel_generalized:estimands}
Beyond providing a test with valid FDR, \cref{lem:conservative-pi} also provides insights on the range of possible effects that can be ascribed to a possible response distribution. Specifically, because the model is not identifiable, one can not estimate a single latent response effect. Rather, \cref{lem:conservative-pi} implies that there is an \emph{interval} of feasible effects as summed up by the following corollary.

\begin{corollary}
\label{corr:np-ite}
Let $f_0(y \mid x)$, $f_t(y \mid x)$, and $\pi^\star(x)$ be as defined in \cref{lem:conservative-pi}. If $f_1$ is a valid responder distribution in \cref{eqn:two_groups}, we must have 
\[ \E_{f_1}[Y|X=x] \in \left[\mu_t(x),  \mu_1^\star(x) \right] \cup \left[\mu_1^\star(x), \mu_t(x) \right], \]
where 
\begin{align*}
\mu_t(x) = \E[Y | X=x, T=1], 
\mu_0(x) = \E[Y | X=x, T=0], \text{ and }
\mu_1^\star(x) = \frac{1}{\pi^\star(x)}(\mu_t(x) - \mu_0(x))  - \mu_0(x).
\end{align*}
Moreover, for any $v \in \left[\mu_t(x),  \mu_1^\star(x)\right] \cup \left[\mu_1^\star(x), \mu_t(x) \right]$, there is a valid responder distribution $f_1$ such that $\E_{f_1}[Y|X=x] = v$.
\end{corollary}
In the above statement, we have used the convention that $[a,b]$ is empty whenever $b<a$. Thus, for any individual $x$, the CARE lies in an interval that is bounded on one side by the CATE value $\mu_t(x) - \mu_0(x)$ and on the other side by the extremal CARE value of $\mu_1^\star(x) - \mu_0(x)$.

\begin{figure*}
\centering
\includegraphics[width=.95\textwidth]{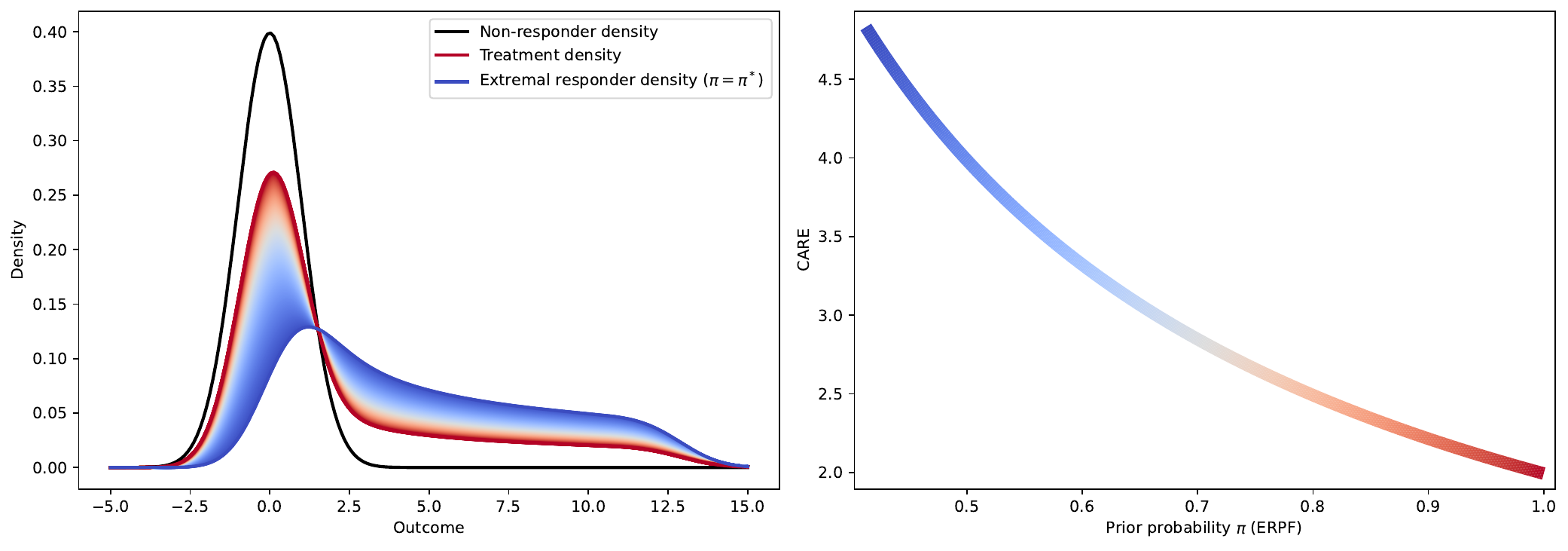}
\caption{Nonparametric C2G example. \emph{Left:} An example of a non-responder density (black), a treatment density (red), an extremal responder density (blue), and the range of feasible responder densities (gradient). \emph{Right:} Each feasible responder density has an associated prior probability and CARE value.}
\label{fig:npc2g-example}
\end{figure*}

\cref{fig:npc2g-example} shows an example of the nonparametric C2G setup: the control density, the treatment density, and the range of feasible responder densities. We see that at one end of the scale, the treatment density itself is a feasible alternative, with an associated prior probability $\pi(x) = 1$ and induced CARE equal to the CATE. On the other end of the scale, there is an extremal responder density with a smaller associated prior probability ($\pi(x) = \pi^\star(x)$) and larger CARE value. The relationship between these quantities is that as the prior probability grows towards one, the CARE shrinks towards the CATE, reflecting the fact that the associated responder distribution must also shrink towards the treatment distribution itself.

\section{Robustness against unmeasured confounding}
\label{sec:confounding}

\begin{figure*}[t]
\centering
\begin{subfigure}{0.22\linewidth}
\centering
\scalebox{0.7}{\begin{tikzpicture}

  \node[obs, ]                               (x) {$X$};
  \node[latent, below=of x]                (h) {$H$};
  \node[obs, left=of h]                   (t) {$T$};
  \node[obs, right=of h]                   (y) {$Y$};
  \node[latent, right=of x]                 (u) {$U$};

  \edge {x} {t} ; %
  \edge {x} {h}  ; %
  \edge {x} {y}  ; %
  \edge {t} {h}  ; %
  \edge {h} {y}  ; %
  \edge {u} {t}  ;
  \edge {u} {y}  ;

\end{tikzpicture}}
\caption{\centering \label{fig:confounding:canonical}Canonical\newline confounding}
\end{subfigure}\quad
\begin{subfigure}{0.22\linewidth}
\centering
\scalebox{0.7}{\begin{tikzpicture}

  \node[obs, ]                               (x) {$X$};
  \node[latent, below=of x]                (h) {$H$};
  \node[obs, left=of h]                   (t) {$T$};
  \node[obs, right=of h]                   (y) {$Y$};
  \node[latent, right=of x]                 (u) {$U$};

  \edge {x} {t} ; %
  \edge {x} {h}  ; %
  \edge {x} {y}  ; %
  \edge {t} {h}  ; %
  \edge {h} {y}  ; %
  \edge {u} {t}  ;
  \edge {u} {h}  ;

\end{tikzpicture}}
\caption{\centering \label{fig:confounding:compliance}Response\newline confounding}
\end{subfigure}\quad
\begin{subfigure}{0.22\linewidth}
\centering
\scalebox{0.7}{\begin{tikzpicture}

  \node[obs, ]                               (x) {$X$};
  \node[latent, below=of x]                (h) {$H$};
  \node[obs, left=of h]                   (t) {$T$};
  \node[obs, right=of h]                   (y) {$Y$};
  \node[latent, right=of x]                 (u) {$U$};

  \edge {x} {t} ; %
  \edge {x} {h}  ; %
  \edge {x} {y}  ; %
  \edge {t} {h}  ; %
  \edge {h} {y}  ; %
  \edge {u} {h}  ;
  \edge {u} {y}  ;

\end{tikzpicture}}
\caption{\centering \label{fig:confounding:effect}Effect\newline confounding}
\end{subfigure}\quad
\begin{subfigure}{0.22\linewidth}
\centering
\scalebox{0.7}{\begin{tikzpicture}

  \node[obs, ]                               (x) {$X$};
  \node[latent, below=of x]                (h) {$H$};
  \node[obs, left=of h]                   (t) {$T$};
  \node[obs, right=of h]                   (y) {$Y$};
  \node[latent, right=of x]                 (u) {$U$};

  \edge {x} {t} ; %
  \edge {x} {h}  ; %
  \edge {x} {y}  ; %
  \edge {t} {h}  ; %
  \edge {h} {y}  ; %
  \edge {u} {t}  ;
  \edge {u} {h}  ;
  \edge {u} {y}  ;

\end{tikzpicture}}
\caption{\centering \label{fig:confounding:total}Total\newline confounding}
\end{subfigure}
\caption{\label{fig:confounding}
The four types of latent confounding in the causal two-groups model.}
\end{figure*}
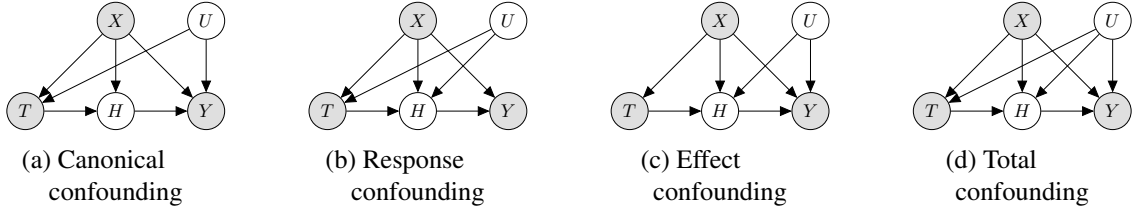

In this section, we consider how latent confounding impacts statistical inferences in the causal two-groups model. Latent confounding in the C2G model occurs when an unobserved variable $U$ affects two or more of the variables in $\{T, H, Y \}$.  \Cref{fig:confounding} shows the four types of possible latent confounding scenarios.

The key question is whether either of the C2G models proposed above remains \emph{well-specified} under any of these latent confounding scenarios. By well-specified, we mean the observed non-responder and treatment distributions are within the class of probability distributions specified by the modeling assumptions. A well-specified model under latent confounding ensures that integrating out the latent confounder in the untreated group yields the original non-responder distribution and in the treatment group yields the original mixture model,
\begin{equation}
\begin{aligned}
\int p(y | x, u, t=0) p(u) du & = f_0(y \mid x) = p(y | x, h=0, t=1) \, , \\
\int p(y | x, u, t=1) p(u) du &= (1-\pi(x)) f_0(y \mid x) + \pi(x) f_1(y \mid x) \, ,
\end{aligned}
\label{eqn:well-specified}
\end{equation}
where $(\pi, f_0, f_1)$ are specified by the modeling assumptions (i.e. parametric, semi-parametric, or nonparametric). If a model remains well-specified in the presence of latent confounding, statistical inferences will still be valid.

\subsection{Confounding in the additive model}
\label{subsec:confounding:additive}

The additive causal two-groups model requires that the $(f_0, f_1)$ distributions obey \cref{eqn:additive_errors}. Because \cref{eqn:additive_errors} is agnostic to the relationship between $T$ and $H$, it is straightforward to show that the additive causal two groups model remains correctly specified and applicable under response confounding.
\begin{proposition}
\label{prop:compliance-confounding-ac2g}
The additive causal two groups model is robust against response confounding.
\end{proposition}

The story changes when we allow an unobserved random variable to affect the output $Y$. Indeed, even when this variable only affects the conditional means of the non-responder and responder distributions by an additive shift while leaving the noise distribution unchanged, the model is misspecified.

\begin{proposition}
\label{prop:latent-nonrobust-ac2g}
The additive causal two-groups model is misspecified in the presence of an unobserved $U$ such that 
\begin{align*}
Y | X=x, U=u, H=h &=& \mu_{h}(x) + \rho_{h}(u) + \epsilon  \\
\epsilon &\sim& g.
\end{align*}
As a consequence, the additive causal two-groups model is not robust in the presence of canonical, effect, or total confounding.
\end{proposition}

\subsection{Confounding in the nonparametric model}
\label{subsec:confounding:nonparametric}

For the nonparametric method, there are only two modeling assumptions. First, the outcome distribution of the untreated samples must match that of the treated non-responders. Second, the outcome distribution of the treated samples can be written as a mixture model of the form given by \cref{eqn:treatment-mixture}.
The following result shows that this minimalist approach offers an added benefit over the additive model in the form of validity under effect confounding as well as response confounding.

\begin{proposition}
\label{prop:confounding-npc2g}
The nonparametric causal two-groups model is robust against response and effect confounding.
\end{proposition}

The reason why the nonparametric causal two-groups model remains well-specified in the presence of these two types of confounding is that both still allow for $Y \indep T \mid X,  H$. When this conditional independence relationship holds, the non-responder outcome distribution coincides with the untreated outcome distribution. This allows us to write the treatment distribution as the mixture in \cref{eqn:treatment-mixture}, thereby satisfying the conditions of \cref{eqn:well-specified}. When we do not have $Y \indep T \mid X,  H$, we cannot guarantee that the nonparametric causal two-groups model is well-specified, as the following result shows.

\begin{proposition}
\label{prop:canonical-confounding}
The nonparametric causal two-groups model is not robust in the presence of canonical or total confounding.
\end{proposition}

The key to proving \cref{prop:canonical-confounding} is in showing that a canonical confounder $U$ can lead to situations in which the untreated outcome distribution $p(Y \mid X, T=0)$ differs from the treated-but-non-responder outcome distribution $p(Y \mid X, H=0, T=1)$, breaking the fundamental assumptions of the nonparametric causal two-groups model.

\section{Simulations}
\label{sec:simulations}

We compare the additive (Add-C2G) and nonparametric (NP-C2G) causal two-groups models against other methods for FDR control. First, we compare against the frequentist baseline that fits the same model as NP-C2G to the untreated (non-responder) distribution, but uses the predicted densities as a null distribution to calculate frequentist $p$-values, with Benjamini-Hochberg correction for multiple comparisons~\citep{benjamini:hochberg:1995:bh-fdr}. Second, we compare against an oracle baseline that uses the NP-C2G methodology with full knowledge of the responder and non-responder distributions. \Cref{sec:more_simulations} contains an expanded set of comparisons against other causal inference methods (FDRregression~\citep{scott:etal:2014:fdr-regression}, Bayesian additive regression trees~\citep{hill:etal:2011:causal-bart}, and causal forests~\citep{wager:athey:2018:causal-forests}).  We focus our comparisons on synthetic simulations; \Cref{sec:gdsc_simulations} contains semi-synthetic simulations on a drug response dataset. Across all simulations, we generally find that the causal two-groups models achieved high power and valid FDR when the appropriate modeling assumptions held. C2G models consistently outperform the frequentist baseline and other causal inference methods, often approaching the performance of the oracle model.

In all synthetic settings, the covariates $X \in \R^d$ are normally distributed with $d=10$, and the distribution of $H| T=1$ is determined by a random logistic regression. In the additive setting, the responder and non-responder distributions are both normal, with means that are some random function of $X$. In the nonadditive setting, the outcome distributions are normally distributed conditioned on both the covariates and some exogenous random variables. For both settings, the difference between the non-responder and responder means is roughly proportional to a parameter $\tau$, which we vary from $1$ to $5$. The number of data points is varied from $N=1$K to $N=10$K. We run each simulation for 50 different random seeds. At each nominal FDR level $\alpha$, we measure performance with respect to empirical FDR and power, defined as the ratio of selected responders over the total number of responders.

\begin{table}[t]
\centering
\begin{tabular}{||c|c|c|c|c|c|c||}
\hline
\multicolumn{7}{|c|}{FDR @ $\alpha=0.1$} \\
\hline
& \multicolumn{3}{|c|}{Additive} & \multicolumn{3}{|c|}{Nonadditive} \\
\hline
Method & $\tau=1$ & $\tau=3$ & $\tau=5$ & $\tau=1$ & $\tau=3$ & $\tau=5$ \\
\hline
Frequentist & \textbf{ 0.04±0.05 } & \textbf{ 0.02±0.01 } & \textbf{ 0.02±0.0 } & \textbf{ 0.02±0.0 } & \textbf{ 0.02±0.0 } & \textbf{ 0.02±0.0 } \\
Add-C2G & \textbf{ 0.13±0.06 } & \textbf{ 0.11±0.02 } & \textbf{ 0.11±0.01 } & \textbf{ 0.08±0.02 } & \textbf{ 0.05±0.01 } & \textbf{ 0.03±0.01 } \\
NP-C2G & \textbf{ 0.06±0.04 } & \textbf{ 0.05±0.01 } & \textbf{ 0.08±0.01 } & \textbf{ 0.07±0.01 } & \textbf{ 0.09±0.01 } & \textbf{ 0.11±0.01 } \\
NP-Oracle & \textbf{ 0.07±0.03 } & \textbf{ 0.1±0.01 } & \textbf{ 0.1±0.0 } & \textbf{ 0.09±0.01 } & \textbf{ 0.09±0.01 } & \textbf{ 0.12±0.03 } \\
\hline
\hline
\multicolumn{7}{|c|}{Power @ $\alpha=0.1$} \\
\hline
& \multicolumn{3}{|c|}{Additive} & \multicolumn{3}{|c|}{Nonadditive} \\
\hline
Method & $\tau=1$ & $\tau=3$ & $\tau=5$ & $\tau=1$ & $\tau=3$ & $\tau=5$ \\
\hline
Frequentist & 0.0±0.0 & 0.22±0.05 & 0.67±0.06 & 0.38±0.06 & 0.57±0.04 & 0.65±0.03 \\
Add-C2G & 0.01±0.01 & 0.48±0.08 & \textbf{ 0.88±0.03 } & 0.45±0.05 & 0.53±0.03 & 0.54±0.03 \\
NP-C2G & 0.0±0.0 & 0.44±0.06 & \textbf{ 0.84±0.05 } & \textbf{ 0.51±0.06 } & \textbf{ 0.68±0.04 } & \textbf{ 0.74±0.03 } \\
NP-Oracle & \textbf{ 0.04±0.01 } & \textbf{ 0.67±0.05 } & \textbf{ 0.92±0.03 } & \textbf{ 0.58±0.05 } & \textbf{ 0.72±0.03 } & \textbf{ 0.79±0.03 } \\
\hline
\end{tabular}
\caption{Multiple testing results on synthetic data for $N=1$K. $\pm$ denotes 95\% confindence intervals. \textbf{Bolded} results in FDR section indicate that the method(s) achieved valid FDR (up to 95\% confidence intervals). \textbf{Bolded} results in power section indicate that the method(s) achieved the highest power for that setting (up to 95\% confidence intervals).}
\label{table:synthetic-simulations}
\end{table}

\begin{table}[t]
\centering
\begin{tabular}{||c|c|c|c|c|c|c||}
\hline
& \multicolumn{3}{|c|}{Additive} & \multicolumn{3}{|c|}{Nonadditive} \\
\hline
N & $\tau=1$ & $\tau=3$ & $\tau=5$ & $\tau=1$ & $\tau=3$ & $\tau=5$ \\
\hline
1K & 0.19±0.03 & 0.63±0.05 & 0.81±0.03 & 0.67±0.05 & 0.66±0.06 & 0.63±0.06 \\
10K & 0.11±0.01 & 0.6±0.04 & 0.85±0.03 & 0.66±0.05 & 0.75±0.06 & 0.69±0.07 \\
\hline
\end{tabular}
\caption{Average response effect (ARE) interval results on synthetic data for NP-C2G. Values indicate Jaccard index with the oracle ARE interval; $\pm$ denotes 95\% confidence intervals. }
\label{table:nonadditive-jaccard-ite}
\end{table}

\Cref{table:synthetic-simulations} displays the results for the synthetic simulations with $N=1$K. Across all settings, NP-C2G (both empirical and oracle versions), Add-C2G, and the frequentist baseline generally control FDR at the level considered. However, the frequentist baseline sacrifices power compared to the other methods. Moreover, in many settings the NP-C2G model is competitive with the oracle model.

We also measure the accuracy of the response effect intervals generated by the NP-C2G model, according to \cref{corr:np-ite}. We compare the intervals with the oracle NP-C2G intervals by averaging the predicted intervals across the population and computing their corresponding Jaccard index. Here the Jaccard index of two finite subsets of the reals $I, J \subset \R$ is given by $\frac{\lambda(I \cap J)}{\lambda(I \cup J)}$, where $\lambda(\cdot)$ denotes Lebesgue measure. \cref{table:nonadditive-jaccard-ite} shows the results for both $N=1$K and $N=10$K. We observe that more data and larger effect sizes generally lead to more accurate interval estimates.

\section{Case study: Immune checkpoint inhibitors}
\label{sec:case-study}

We apply the NP-C2G model to a dataset tracking survival of cancer patients treated with immune checkpoint inhibitors~\citep{samstein:etal:2019:tmb-immunotherapy-survival}. Each patient in the study had targeted next-generation sequencing performed on their tumors, recording the mutational state of 341-468 genes. The sequencing data generates a binary vector for each patient where one indicates a somatic mutation in a gene and zero indicates the gene is not mutated. Tumors were sequenced with different gene panels, leaving missing data in the patient-by-gene matrix. To handle this, we performed Bernoulli matrix factorization on the sequencing matrix to create 5-dimensional embeddings for each patient, where the dimension was chosen using 3-fold cross-validation on held out log-likelihood. The embeddings, combined with the clinical features of sex, age, cancer type and metastatic status, form the confounders in our analysis.

The study tracked which of three treatments was administered: anti-CLT4A, anti-PD-1/PD-L1, or a combination of both. For our analysis, we consider patients receiving anti-CLT4A drugs ($n=99$) as the baseline (untreated) group and patients receiving combination treatments ($n=255$) as the treatment group. A responder in our study is therefore a patient who received an additional effect from PD-1/PD-L1 therapy on top of the baseline effect from anti-CTLA4 therapy. Determining which patients respond to combination therapy versus monotherapy is valuable because immunotherapies can trigger life-threatening complications. A stratification of patients into those likely to benefit or be harmed by combination therapy could inform future therapy decisions and potentially save lives.

\begin{figure*}
\centering %
\includegraphics[width=.95\textwidth]{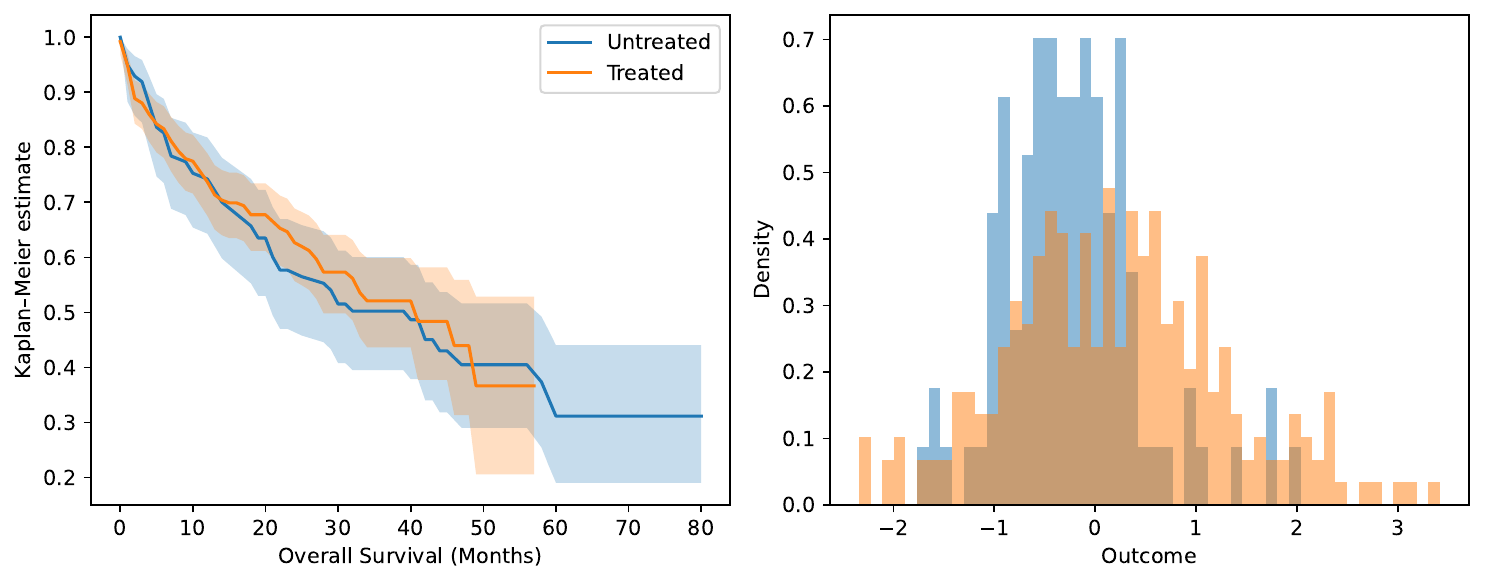}
\caption{Case study survival data. \emph{Left:} Kaplan-Meier fits to the survival data from the untreated and treated groups. \emph{Right:} Histogram of Cox proportional hazard-transformed outcomes for the untreated and treated groups.}
\label{fig:case-study-data}
\end{figure*}

To transform the survival information to numerical outcomes, we first fit a Cox proportional hazards model to the monotherapy group. For each patient $i$, we let $s_i$ denote the patient's survival time if they were uncensored. If they were censored, $s_i$ is the expectation of this time under the Cox model conditioned on being at least the observed survival time. This provides a conservative estimate of the true survival time because the conditional expectation will shrink survival times to the mean. We convert survival times to $z$-scores via the inverse normal cumulative distribution function $P(\text{survival time} \leq s_i)$ under the Cox model. We use the $z$-score as the outcome variable $y_i$ for each patient. \Cref{fig:case-study-data} shows the survival curves for both groups and the resulting outcome variables.

We fit the NP-C2G model according to the methodology outlined in \cref{sec:kernel_generalized}. For a fine-grained grid of nominal FDR thresholds $\alpha \in [0, 0.25]$, we calculate the rejection set of treated patients $S_\alpha$. To detect biomarkers of response, we perform a hypothesis test for each gene in the dataset. We compute Fisher's exact test on the contingency table over treated patients where the two variables are whether or not the patient had the mutated gene and whether or not the patient was in the rejected set $S_\alpha$. We consider both possible alternatives in our test: (i) a randomly generated table has a larger number of patients that are in $S_\alpha$ and do not have the corresponding mutation, and (ii) a randomly generated table has a smaller number of patients that are in $S_\alpha$ and do not have the corresponding mutation. After calculating the p-values for all the genes under (i) or (ii), we performed Benjamini-Hochberg correction to obtain q-values. We consider a range of possible thresholds to reject and record those genes that are rejected with q-values less than 0.1. Genes rejected under alternative (i) are termed \emph{more favorable}, as it is unlikely to have more treated patients that are both in $S_\alpha$ and have the mutation. Analogously, genes rejected under alternative (ii) are termed \emph{less favorable}.

For each detected biomarker, we report the two estimands of interest from \cref{subsec:model:estimands}. For the conditional average response effect, we calculate the upper and lower bounds of the individual response effects given by \cref{corr:np-ite}, and then average these values over individuals carrying a mutation in a given gene. For the expected responder population fraction, we use the conservative prior $\pi^\star$ to generate a conservative estimate of the responder fraction. Each calculation is performed conditioned on the mutated subpopulation and thus we report the conditional expected responder population fraction (CERPF).

\cref{table:case study} displays the results for the selected genes. Among the mutations identified as more favorable, VHL mutations have the smallest rejection threshold, the smallest q-values, the largest CERPF, and a CARE interval that is entirely positive. VHL mutations have been implicated in the literature as a positive biomarker for anti-PD-1/PD-L1 therapies~\citep{meng:etal:2024:vhl-loss-antitumor-immunity}, suggesting this is likely a true positive. Among the mutations identified as less favorable, BRAF mutations have the smallest rejection threshold, the smallest q-values, and have been implicated in increased toxicity of combination anti-CTLA4/anti-PD-1 treatments~\citep{piresdasilva:etal:2022:braf-toxicity-pd1}. Among the other mutations, there is evidence suggesting SETD2 mutations lead to more favorable outcomes for immunotherapy~\citep{jee:etal:2024:msk-chord,lu:etal:2021:SETD2-immunotherapy} and TERT mutations are associated with less favorable outcomes for anti-PD-1 therapies compared with anti-CTLA4 therapies~\citep{li:etal:2020:tert-immunotherapy}. The results for ERBB4 and GRIN2A represent potential novel discoveries. Given their small CERPF, it is likely they only affect a small subset of patients and may warrant further investigation to find more fine-grained molecular markers of response.

\begin{table}
\centering
\begin{tabular}{|c|c|c|c|c|c|c|}
\hline
& \multicolumn{2}{|c|}{More favorable} & \multicolumn{4}{|c|}{Less favorable} \\
\hline
Mutation & SETD2 & VHL & BRAF & ERBB4 & GRIN2A & TERT \\
\hline
Smallest $\alpha$ & 0.07 & 0.029 & 0.149 & 0.239 & 0.247 & 0.201 \\
\hline
Largest -$\log_{10}(q)$ & 1.577 & 7.511 & 3.373 & 1.125 & 1.003 & 1.963 \\
\hline
CARE interval & (-4.3, 1.4) & (1.3, 2.1) & (-8.5, 0.2) & (-9.0, 1.6) & (-7.8, 1.1) & (-7.2, 1.2) \\
\hline
CERPF & 0.39 & 0.62 & 0.22 & 0.16 & 0.15 & 0.25\\
\hline
\end{tabular}
\caption{Case study results comparing outcomes from CTLA4 treatment (null) and CTLA4/PD-1 combo treatment (treatment). \emph{Smallest $\alpha$} denotes the smallest nominal FDR rate $\alpha$ of the NP-C2G procedure that results in the gene being selected. \emph{Largest $- \log_{10}(q)$} denotes the largest $q$-value observed for the selected gene. \emph{CARE interval} denotes the average of the treatment effect intervals produced by NP-C2G over the population with the given mutation. \emph{CERPF} denotes the average of the lower-bound prior probabilities produced by NP-C2G over the treated population with the given mutation.}
\label{table:case study}
\end{table}

\section{Discussion}

Modern precision medicine is a practice of gradual refinement. Treatments are delivered first to a broad group of patients, results are measured, and correlates of response are used to stratify patients into subgroups. For instance, breast cancer patients initially received chemotherapy until, in 1967, two hormone receptors were discovered as predictive biomarkers in subsets of patients~\citep{jensen:etal:1981:er-pr-breast-cancer-discovery}. The introduction of drugs blocking these hormones led to improved outcomes for hormone-positive patients~\citep{ward:etal:1973:breast-er-trial,osborne:etal:1980:value-er-pr-breast}. A third marker, the surface protein HER2, was discovered in 1987~\citep{slamon:etal:1987:her2-breast-cancer-discovery}. More recently, clinical trials in immunotherapy have led to approved treatments for so-called triple-negative breast cancers (those without any of the three previously established biomarkers) that are positive for the PD-L1 biomarker~\citep{emens:etal:2021:tnbc-nab-pac-atezo-trial}. At each step along the way, patient data was retrospectively analyzed to identify responders and stratify by a predictive marker.

The emergence of large, real-world datasets of electronic health records~\citep[e.g.][]{jee:etal:2024:msk-chord} presents both an opportunity and a complication for data analysts in precision medicine. Cohorts of tens of thousands of patients are now readily available for mining correlations of treatment response, adverse events, overall survival, and other outcomes of interest. Analysts now can detect significant associations in fine-grained subsets of patients to stratify patients to a near-personalized level of precision. The trade-off is that these are not randomized clinical trial data. They are susceptible to all the usual pitfalls and problems of causal inference in observational data. Advancing statistical analyses for the modern era requires new causal methodology for refining patient cohorts. The methods presented here represent a step in this direction.

Going forward, the causal two-groups model raises some interesting statistical questions. On the theoretical side, one might wonder how to give a precise delineation between identifiable and non-identifiable C2G models. In this work, we gave two examples of modeling assumptions that lead to identifiable C2G models: certain classes of parametric models and additive semi-parametric models. But this is far from a complete characterization. What are the conditions that need to be placed on a set of modeling assumptions to ensure that the resulting C2G model is identifiable?

On the practical side, one might think about extending the causal two-groups model beyond binary treatments. Indeed, many types of interventions are not binary: drugs have doses, treatments have schedules, jobs training programs have multiple days of courses, and diet interventions involve multiple foods and meals. In all of these settings, however, we can still ask the question of whether or not an individual responded to their associated treatment. What then is the appropriate extension of the C2G framework that allows for different degrees or types of treatment?

One might also ask what role independence plays in the C2G model. As it turns out, the FDR control results (\cref{thm:additive-fdr} and \cref{thm:general-fdr}) do not require i.i.d. observations. They only require accurate estimates of the outcome distributions. However, our approach to acquiring those estimates makes use of assumed i.i.d. structure among individuals. In what ways can we relax these i.i.d. assumptions in the C2G model without sacrificing statistical guarantees? Is it possible to consider interference effects and other non-i.i.d. confounding?

Finally, our immunotherapy case study presents a potentially new way to analyze combination therapy studies. Rather than simply considering the difference in survival time, one can now estimate how many patients actually benefit from combination therapy. Application to datasets of different treatment regimens may reveal that certain combinations are effective only for a small group of patients who have an outsized effect, suggesting the need for better patient stratification. %

\subsection*{Acknowledgements}

We gratefully acknowledge support by the NIH/NCI (R37 CA271186, U54 CA274492, P30 CA008748), Break Through Cancer, and the Maurice Campbell Initiative at Memorial Sloan Kettering Cancer Center.

{
\small
\bibliography{main}
}

\newpage
\appendix

\section{Additional results and details from \Cref{sec:model}}
\subsection{Proof of \cref{thm:unident}}
Consider the following densities over $[0,\infty)$: 
\begin{align*}
    f_0(y) &= \frac{3}{2} e^{-y} - e^{-2y}\\
    f^{(1)}_1(y) &= e^{-y} \\
    f^{(2)}_1(y) &= 2 e^{-2y}.
\end{align*}
It can be readily verified that these are all valid probability densities. Moreover, we can see that $f_0$ is log-concave over its domain since
\[ \frac{d^2}{d y^2} \log f_0(y) = -\frac{2e^y}{3(\frac{2}{3} - e^y)^2} < 0 \, \, \text{ for all 
 } y \geq 0. \]
$f^{(1)}_1$ and $f^{(2)}_1$ are also clearly log-concave since their log-densities are linear.

Suppose that $f_0$ is the observed null distribution. Letting $\pi_1 = \frac{3}{4}$ and $\pi_2 = \frac{1}{4}$, consider the following two observed treatment distributions:
\begin{align*}
f^{(1)}_t(y) &= (1-\pi_1) f_0(y) + \pi_1 f^{(1)}_1(y) \\
f^{(2)}_t(y) &= (1-\pi_2) f_0(y) + \pi_2 f^{(2)}_1(y).
\end{align*}
Using these definitions, we have
\begin{align*}
f^{(1)}_t(y) &= (1-\pi_1) f_0(y) + \pi_1 f^{(1)}_1(y) \\
&= \frac{1}{4} \left( \frac{3}{2} e^{-y} - e^{-2y} \right) + \frac{3}{4} e^{-y} \\
&= \frac{9}{8} e^{-y} - \frac{1}{4} e^{-2y} \\
&= \frac{3}{4} \left( \frac{3}{2} e^{-y} - e^{-2y} \right) + \frac{1}{4} \cdot 2 e^{-y} \\
&= (1-\pi_2) f_0(y) + \pi_2 f^{(2)}_1(y) \\
&= f^{(2)}_t(y).
\end{align*}
Thus, we can conclude that the C2G model is unidentifiable even when both the null distribution and the alternative distribution are restricted to being log-concave.

\subsection{Proof of \cref{prop:normal-identifiability}}

Fix $x$. In \Cref{eqn:two_groups}, we observe distributions $f_0(y \mid x)$ and $f_t(y \mid x)$. Suppose we are given $\mu_0 = \mu_0(x)$, $\mu_1 = \mu_1(x)$, $\mu_2 = \mu_2(x)$, $\sigma_0^2 = \sigma_0^2(x)$, $\sigma_1^2 = \sigma_1^2(x)$, $\sigma_2^2 = \sigma_2^2(x)$,  $\pi_1 = \pi_1(x)$, and $\pi_2 = \pi_2(x)$, where $\pi_1, \pi_2 \in (0,1)$. By translating and scaling the space, we can assume without loss of generality that $\mu_0 = 0$ and $\sigma_0^2 = 1$. Our goal is to show that if
\begin{align}
\label{eqn:normal-identifiability}
(1-\pi_1)\Ncal(y \mid 0, 1) + \pi_1 \Ncal(y \mid \mu_1, \sigma_1^2) = (1-\pi_2)\Ncal(y \mid 0, 1) + \pi_2 \Ncal(y \mid \mu_2, \sigma_2^2) 
\end{align}
for all $y \in \R$, then we must have $\pi_1=\pi_2$, $\mu_1 = \mu_2$, and $\sigma_1^2 = \sigma_2^2$. There are two cases to consider here.

\paragraph{Case 1: $\pi_1 = \pi_2$.} In this case, we immediately have
\[ \Ncal(y \mid \mu_1, \sigma_1^2) = \Ncal(y \mid \mu_2, \sigma_2^2) \]
for all $y \in \R$. Taking the first moment of both sides, we have $\mu_1 = \mu_2$. Conditioned on this fact, we can calculate the variance of both sides to conclude $\sigma_1^2 = \sigma_2^2$. 

\paragraph{Case 2: $\pi_1 \neq \pi_2$.} In this case, we will come to a contradiction. Assume without loss of generality that $\pi_1 > \pi_2$. Letting $\beta = \frac{\pi_2}{\pi_1 - \pi_2} > 0$, we can rearrange \cref{eqn:normal-identifiability} to get
\[ \Ncal(y \mid 0, 1) = (1+\beta) \Ncal(y \mid \mu_1, \sigma_1^2) - \beta \Ncal(y \mid \mu_2, \sigma_2^2).  \]
Multiplying both sides by $e^{ty}$ and integrating $y \in \R$, the moment generating function of the normal distribution gives us
\begin{align}
\label{eqn:normal-identifiability-mgf}
\exp\left( t^2/2 \right)
= (1+\beta)\exp\left( \mu_1 t + \sigma_1^2 t^2/2 \right) - \beta \exp\left( \mu_2 t + \sigma_2^2 t^2/2 \right)   
\end{align}
for all $t \in \R$. We will analyze \cref{eqn:normal-identifiability-mgf} in the limit as $t\rightarrow \infty$.

Observe first that we must have $\sigma_2^2 \leq \sigma_1^2$, since otherwise as $t\rightarrow \infty$, the RHS of \cref{eqn:normal-identifiability-mgf} tends to $-\infty$ while the LHS is positive. 

Now consider the case where $\sigma_2^2 < \sigma_1^2$. Multiplying \cref{eqn:normal-identifiability-mgf} through by $e^{-\mu_1 t - \sigma_1^2 t^2/2}$, we have
\[ \exp\left( -\mu_1 t - (\sigma_1^2 - 1) t^2/2 \right) = 1+\beta - \beta \exp\left( (\mu_2 - \mu_1) t - (\sigma_1^2 -\sigma_2^2) t^2/2 \right).  \]
As $t\rightarrow \infty$, the RHS tends to $1 + \beta$. On the other hand, the LHS must tend to either 0 or $+\infty$. Thus we must have $\sigma_1^2 = \sigma_2^2$. In which case, we can rearrange \cref{eqn:normal-identifiability-mgf} to obtain
\[ 1 = \left[(1+\beta)\exp(\mu_1 t) - \beta \exp(\mu_2 t) \right] \exp\left( (\sigma_1^2 - 1)t^2/2 \right). \]
If $\sigma_1^2 < 1$, then the RHS must tend to 0. If $\sigma_1^2 > 1$, then the RHS must tend to $\pm \infty$. Thus, we can conclude $\sigma_1^2 = 1$, and the above equation becomes
\[ 1 = \left[(1+\beta)\exp(\mu_1 t) - \beta \exp(\mu_2 t) \right]. \]
The only way for this to hold for all $t$ is if $\mu_1 = \mu_2=0$.

\section{Additional results and details from \Cref{sec:additive}}
\subsection{Proof of \cref{thm:additive-identify}}

For any $x$, observe that we can shift the space by $\mu_0(x)$ so that the model in \Cref{eqn:additive_errors} becomes
\begin{align*}
f_0(y \mid x) &= g(y) \\
f_1(y \mid x) &= g(y - \tau(x)).
\end{align*}
In this case, the observable distributions are the null distribution $f_0 = g$ and the treatment distribution
\[ f_t(y \mid x) = (1-\pi(x)) g(y) + \pi(x) g(y - \tau(x)).  \]
Fixing a particular $x$, the identifiability of \Cref{eqn:additive_errors} can be restated as follows.

\begin{theorem}[Restatement of \cref{thm:additive-identify}]
\label{thm:additive-identify-restate}
Let $g$ denote a probability density function over $\R$. For $\pi \in (0,1)$, $\tau \in \R$, define the density function 
\[ f(y ; \pi, \tau) = (1- \pi) g(y) + \pi g(y-\tau). \]
Then the family of densities $\{f(y ; \pi, \tau) \,  \mid \, \pi \in (0,1), \tau \in \R \setminus \{ 0\} \}$ is identifiable.
\end{theorem}
    
We begin by recalling the characteristic function of a probability distribution. If $g$ is a probability distribution over $\R$, then its characteristic function is given by
\[ \varphi_g(\omega) = \E_{X\sim g}\left[ e^{-i \omega X} \right] = \int_{-\infty}^\infty g(x) e^{-i \omega x} \, dx,  \]
where $\omega \in \R$ and the integral is only well-defined when $g$ is a probability density function. We will make use of the following elementary property of the characteristic function.

\begin{lemma}
\label{lem:charfun-property}
For any probability density function $g$, there exists a value $\omega_0 > 0$ such that $|\varphi_g(\omega)| > 0$ for all $\omega \in [0, \omega_0]$.
\end{lemma}
\begin{proof}
This follows directly from the fact that $\varphi_g$ is uniformly continuous on $\R$ and $\varphi_g(0) = 1$ \citep{resnick:2013:probability}.
\end{proof}

With this tool in hand, we turn to the proof of \Cref{thm:additive-identify-restate}.

\begin{proof}[Proof of \Cref{thm:additive-identify-restate}]
Let $\pi_1, \pi_2 \in (0,1)$ and $\tau_1, \tau_2 \in \R \setminus \{ 0 \}$ be parameters such that $f(y;\pi_1, \tau_1) = f(y;\pi_2, \tau_2)$. Our goal is to show that we must have $\pi_1 = \pi_2$ and $\tau_1 = \tau_2$. We consider two cases.

\paragraph{Case 1: $\pi_1 = \pi_2$.} In this case, we must have
\begin{align*}
0 &= f(y;\pi_1, \tau_1) - f(y;\pi_2, \tau_2) \\
&= (1- \pi_1) g(y) + \pi_1 g(y-\tau_1) - (1- \pi_2) g(y) - \pi_2 g(y-\tau_2) \\
&= \pi_1 g(y-\tau_1) - \pi_1 g(y - \tau_2)
\end{align*}
for all $y \in \R$. Multiplying through by $e^{i \omega y}/\pi_1$ and integrating over $\R$, we have
\begin{align*}
0 &= \int_{-\infty}^\infty g(y-\tau_1) e^{i \omega y} \, dy - \int_{-\infty}^\infty g(y-\tau_2) e^{i \omega y} \, dy \\
&= \int_{-\infty}^\infty g(y) e^{i \omega (y+\tau_1)} \, dy - \int_{-\infty}^\infty g(y) e^{i \omega (y+\tau_2)} \, dy \\
&= \varphi_g(\omega) e^{i\omega \tau_1} - \varphi_g(\omega) e^{i\omega \tau_2},
\end{align*}
which holds for all $\omega \in \R$. Invoking \Cref{lem:charfun-property}, there is a value a value $\omega_0 > 0$ such that $|\varphi_g(\omega)| > 0$ for all $\omega \in [0, \omega_0]$. Plugging $\omega = \min \left( \omega_0, \frac{\pi}{|\tau_1 - \tau_2|}\right)$ into the above and dividing by $\varphi_g(\omega)$, we can rearrange to observe
\[ e^{i \omega (\tau_1 - \tau_2)} = 1. \]
This can only be true if $\omega (\tau_1 - \tau_2) = 2 \pi n$ for some $n \in \Z$. By our choice of $\omega$, it must be the case that $\tau_1 = \tau_2$.

\paragraph{Case 2: $\pi_1 \neq \pi_2$.} Here we will reach a contradiction. Assume without loss of generality that $\pi_1 < \pi_2$. Then we have
\begin{align*}
0 &= f(y;\pi_1, \tau_1) - f(y;\pi_2, \tau_2) \\
&= (1- \pi_1) g(y) + \pi_1 g(y-\tau_1) - (1- \pi_2) g(y) - \pi_2 g(y-\tau_2) \\
&= (\pi_2 - \pi_1) g(y) + \pi_1 g(y-\tau_1) - \pi_2 g(y-\tau_2)
\end{align*}
for all $y \in \R$. Multiplying through by $e^{i \omega y}$ and integrating over $\R$, we have
\begin{align*}
0 &= (\pi_2 - \pi_1) \int_{-\infty}^\infty g(y) e^{i \omega y} \, dy + \pi_1\int_{-\infty}^\infty g(y-\tau_1) e^{i \omega y} \, dy - \pi_2 \int_{-\infty}^\infty g(y-\tau_2) e^{i \omega y} \, dy \\
&= (\pi_2 - \pi_1) \varphi_g(\omega) + \pi_1 \varphi_g(\omega) e^{i\omega \tau_1} - \pi_2 \varphi_g(\omega) e^{i\omega \tau_2},
\end{align*}
for all $\omega \in \R$. Invoking \Cref{lem:charfun-property} again, we have $|\varphi_g(\omega)| > 0$ for all $\omega \in [0, \omega_0]$.  Taking $\omega = \min( \omega_0, \frac{\pi}{|\tau_1|})$, we divide through by $\varphi_g(\omega) \pi_2$ and rearrange to get
\[ e^{i \omega \tau_2} = 1 - \frac{\pi_1}{\pi_2} + \frac{\pi_1}{\pi_2} e^{i \omega \tau_1} = 1 - p + p e^{i \omega \tau_1},\]
where $p \in (0, 1/2)$. Taking magnitudes of both sides, we have
\begin{align*}
1 = \left| e^{i \omega \tau_2} \right| &= \left| 1 - p + p e^{i \omega \tau_1} \right| \\
&= \left| 1 - p + p( \cos(\omega \tau_1) + i \sin(\omega \tau_1)) \right| \\
&= \sqrt{ (1 - p + p \cos(\omega \tau_1))^2 + p^2 \sin^2(\omega \tau_1) } \\
&= \sqrt{ 1 - 2p(1-p)(1-\cos(\omega \tau_1))} < 1,
\end{align*}
where the last line follows from our choice of $\omega$. Thus, we have reached a contradiction.
\end{proof}

\subsection{Proof of \Cref{thm:additive-fdr}}

We require the following result showing that we our posterior estimates are accurate.

\begin{lemma}
\label{lem:additive-posterior-approximation}
Suppose that the conditions of \Cref{thm:additive-fdr} hold. Then for all $i=1,\ldots, n$, the absolute difference
\[ \left|\frac{(1-{\pi}(x_i)) {g}(y_i - \mu_0(x_i))}{(1-{\pi}(x_i)) {g}(y_i - \mu_0(x_i)) + {\pi}(x_i) {g}(y_i - \mu_1(x_i))} - \frac{(1-\hat{\pi}(x_i)) \hat{g}(y_i - \hat{\mu_0}(x_i))}{(1-\hat{\pi}(x_i)) \hat{g}(y_i - \hat{\mu_0}(x_i)) + \hat{\pi}(x_i) \hat{g}(y_i - \hat{\mu_1}(x_i))} \right| \]
is bounded above by $\frac{2U(U+2(\lambda + 1))}{L^2} \epsilon. $
\end{lemma}
\begin{proof}
Pick an index $i=1, \ldots, n$. We use the following notation
\begin{align*}
g_0 = {g}(y_i - \mu_0(x_i)), \, 
g_1 = {g}(y_i - \mu_1(x_i)), \,
\hat{g_0} = \hat{g}(y_i - \hat{\mu_0}(x_i)), \,
\hat{g_1} = \hat{g}(y_i - \hat{\mu_1}(x_i)).
\end{align*}
The smoothness of $g$ combined with the approximation factor of $\hat{\mu_0}$ implies
\[ |g_0 - \hat{g_0}| \leq |{g}(y_i - \mu_0(x_i)) - {g}(y_i - \hat{\mu_0}(x_i))| + |{g}(y_i - \hat{\mu_0}(x_i)) - \hat{g}(y_i - \hat{\mu_0}(x_i))| \leq \epsilon(\lambda + 1).\]
A symmetrical argument shows the same holds for $|g_1 - \hat{g_1}|$. Thus, we can conclude the following
\begin{align*}
|\pi(x_i) g_1 - \hat{\pi}(x_i)\hat{g}_1| &\leq \epsilon(U + 2(\lambda + 1) ) \\
|(1-\pi(x_i)) g_0 - (1-\hat{\pi}(x_i))\hat{g}_0 | &\leq \epsilon(U + 2(\lambda + 1) ) \\
(1-\hat{\pi}(x_i))\hat{g}_0  + \hat{\pi}(x_i)\hat{g}_1 &\geq L/2 \\ 
(1-{\pi}(x_i)){g}_0  + {\pi}(x_i){g}_1 &\leq U.
\end{align*}
Using the shorthand $a = (1-\pi(x_i)) g_0$, $\hat{a} = (1-\hat{\pi}(x_i))\hat{g}_0$, $b = \pi(x_i) g_1$, and $\hat{b} = \hat{\pi}(x_i)\hat{g}_1$, we have
\begin{align*}
\left| \frac{a}{a+b} - \frac{\hat{a}}{\hat{a} + \hat{b}} \right| 
&= \left|\frac{\hat{a}b - \hat{b}a}{(a+b)(\hat{a} + \hat{b})} \right| \\
&\leq \frac{2}{L^2}\left|\hat{a}b - \hat{b}a\right| \\
&\leq \frac{2}{L^2} (a+b) \epsilon (U + 2) = \frac{2U(U+2(\lambda + 1))}{L^2} \epsilon. \qedhere
\end{align*}

\end{proof}

Our next lemma is an elementary result concerning sums of probabilities.
\begin{lemma}
\label{lem:prob-subset-sum}
Let $p_1, \hat{p}_1, p_2, \hat{p_2}, \ldots, p_n, \hat{p}_n, \alpha, \delta \in [0,1]$ such that $\left|p_i - \hat{p_i}\right| \leq \delta$ for all $i=1,\ldots, n$. Let $V, \hat{V} \subset \{1,\ldots, n\}$ be the largest subsets such that 
\[ \frac{1}{|V|} \sum_{i\in V} p_i \leq \alpha - \delta \text{ and }  \frac{1}{|\hat{V}|} \sum_{i\in \hat{V}} \hat{p}_i \leq \alpha .  \]
Then 
\[\sum_{i\in \hat{V}} 1 - {p}_i \geq \frac{1 - \alpha - \delta}{1 - \alpha + \delta + 1/(|V|+1)}\sum_{i\in {V}} 1 - {p}_i.\]
\end{lemma}
\begin{proof}
Observe first that $|\hat{V}| \geq |V|$, since 
\[ \frac{1}{|{V}|} \sum_{i\in {V}} \hat{p}_i \leq \frac{1}{|{V}|} \sum_{i\in {V}} ({p}_i + \delta) \leq \alpha. \]
Now observe that
\[ \sum_{i\in \hat{V}} 1 - {p}_i \geq \sum_{i\in \hat{V}} 1 - \hat{p}_i - \delta = |\hat{V}|(1 - \alpha - \delta) \geq |{V}|(1 - \alpha - \delta). \]
On the other hand, by considering adding a single element to $V$, we have
\[  \frac{1}{|V|} \sum_{i \in V} p_i \geq \alpha - \delta - \frac{1}{|V|+1}.  \]
Thus,
\[ \sum_{i\in {V}} 1 - {p}_i \leq |V|\left( 1 - \alpha + \delta + \frac{1}{|V|+1} \right). \]
Taking ratios completes the proof.
\end{proof}

We now turn to the proof of \cref{thm:additive-fdr}.

\begin{proof}[Proof of \cref{thm:additive-fdr}]
Let $(x_1, y_1, t_1, h_1), \ldots, (x_n, y_n, t_n, h_n)$ denote the samples from \cref{eqn:additive_errors} (where $h_1,\ldots, h_n$ are unobserved). From the structure of \cref{eqn:additive_errors}, we have that, conditioned on $(x_1, y_1, t_1), \ldots, (x_n, y_n, t_n)$, the random variables $\ind[h_i = 0]$ are Bernoulli random variables with biases $p_i = \pr(h_i = 0 \mid x_i, y_i, t_i)$.

From \Cref{lem:additive-posterior-approximation}, we have that
\[ \left| \hat{w}_i - \pr(H=0 \mid x_i, y_i, t_i=1)\right| \leq \frac{2U(U+2(\lambda + 1))}{L^2} \epsilon \]
for all $i=1,\ldots, n$. If $S$ is a set of indices satisfying $t_i=1$ for all $i \in S$ and $\frac{1}{|S|} \sum_{i\in S} \hat{w}_i \leq \alpha$, then
\begin{align*}
\E \left[ \frac{1}{|S|} \sum_{i \in S} \ind[h_i = 0] \mid (x_1, y_1, t_1), \ldots, (x_n y_n, t_n) \right] 
&= \frac{1}{|S|} \sum_{i \in S} \E\left[ \ind[h_i = 0] \mid (x_1, y_1, t_1), \ldots, (x_n y_n, t_n) \right] \\
&=  \frac{1}{|S|} \sum_{i \in 
S} \hat{w}_i +\frac{2U(U+2(\lambda + 1))}{L^2} \epsilon \\
&\leq \alpha + \frac{2U(U+2(\lambda + 1))}{L^2} \epsilon. 
\end{align*}

To prove the power statement, observe that the Bayes' optimal procedure is to choose the largest set $V$ such that $ \frac{1}{|V|}\sum_{i \in V} p_i \leq \alpha$. Then the power conclusion follows directly from \cref{lem:prob-subset-sum} with $\delta = \frac{2U(U+2(\lambda + 1))}{L^2} \epsilon$.
\end{proof}

\section{Additional results and details from \Cref{sec:kernel_generalized}}
\subsection{Proof of \cref{lem:conservative-pi}}
Suppose there exist distributions $\pi$, $f_0$, $f_1$ such that
\begin{align*}
p(y \mid x, t=0) &:= f_0(y \mid x), \text{ and } \\
p(y \mid x, t=1) &:= f_t(y \mid x) = (1-\pi(x)) f_0(y \mid x) + \pi(x) f_1(y \mid x).
\end{align*}
Pick any $x$. Observe that if $\pi(x)=0$, then the corresponding $\pi^\star(x) = 0$ and the lemma conclusion trivially holds. Thus, we may assume $\pi(x) > 0$, in which case we have by rearrangement:
\[ f_1(y \mid x) = \frac{1}{\pi(x)} \left( f_t(y \mid x) - f_0(y \mid x) \right) + f_0(y \mid x). \]
As a probability density, $f_1(y \mid x) \geq 0$, which implies
\[  \pi(x) \geq 1 - \frac{f_t(y \mid x)}{f_0(y \mid x)}. \]
Since the above holds for all $y$, we therefore must have 
\[  \pi(x) \geq 1 - \min_y \frac{f_t(y \mid x)}{f_0(y \mid x)} = \pi^\star(x). \]

To prove the second part of the lemma, let $x$ be a point and $\hat{\pi}$ be a function satisfying $1 \geq \hat{\pi}(x) \geq \pi^\star(x) > 0$. Then define
\[ f_1^{\hat{\pi}}(y \mid x) = \frac{1}{\hat{\pi}(x)} \left( f_t(y \mid x) - f_0(y \mid x) \right) + f_0(y \mid x). \]
By definition of $\pi^\star$, we have $f_1^{\hat{\pi}}(y \mid x) \geq 0$. Moreover, integrating the right hand side over $y$ shows that $f_1^{\hat{\pi}}(y \mid x)$ integrates to 1. Thus, $f_1^{\hat{\pi}}(y \mid x)$ is a proper density. Rearrangement shows that
\[ f_t(y \mid x) = (1-\hat{\pi}(x)) f_0(y \mid x) + \hat{\pi}(x) f_1^{\hat{\pi}}(y \mid x). \]

\subsection{Proof of \Cref{thm:general-fdr}}

To prove \Cref{thm:general-fdr} we first need to prove some lemmas. The first lemma shows that under the conditions of \Cref{thm:general-fdr},  we can estimate $\pi^\star$ accurately.

\begin{lemma}
\label{lem:w_star-estimate}
Let $\epsilon, L, U > 0$ be given. Suppose the following holds for all $i=1, \ldots n$ and $y \in \Ycal$:
\begin{itemize}
    \item[(i)] $L \leq f_t(y \mid x_i), f_0(y \mid x_i) \leq U$,
    \item [(ii)] $|f_t(y \mid x_i) - \hat{f}_t(y \mid x_i)| \leq \epsilon$, and
    \item[(iii)] $|f_0(y \mid x_i) - \hat{f}_0(y \mid x_i)| \leq \epsilon$.
\end{itemize}
If $\epsilon \leq L/2$, then $|\hat{\pi}^\star(x_i) - \pi^\star(x_i)| \leq \frac{8 U \epsilon}{L^2}$ for all $i=1,\ldots, n$, where
$ \hat{\pi}^\star(x) = 1 - \min_y \frac{\hat{f}_t(y \mid x)}{\hat{f}_0(y \mid x)} .$
\end{lemma}
\begin{proof}
Pick some index $i$, and let $x = x_i$. Let $y = \argmin_y \frac{f_t(y \mid x)}{f_0(y \mid x)}$ and $\hat{y} = \argmin_y \frac{\hat{f}_t(y \mid x)}{\hat{f}_0(y \mid x)}$. Then we can write
\begin{align*}
\hat{\pi}^\star(x) - \pi^\star(x)
&=  \frac{f_t(y \mid x)}{f_0(y \mid x)} - \frac{\hat{f}_t(\hat{y} \mid x)}{\hat{f}_0(\hat{y} \mid x)}  \\
&= \frac{f_t(y \mid x)}{f_0(y \mid x)} - \frac{f_t(\hat{y} \mid x)}{f_0(\hat{y} \mid x)} + \frac{f_t(\hat{y} \mid x)}{f_0(\hat{y} \mid x)} - \frac{\hat{f}_t(\hat{y} \mid x)}{\hat{f}_0(\hat{y} \mid x)} \\
&\leq \frac{f_t(\hat{y} \mid x)}{f_0(\hat{y} \mid x)} - \frac{\hat{f}_t(\hat{y} \mid x)}{\hat{f}_0(\hat{y} \mid x)} \\
&= \frac{f_t(\hat{y} \mid x) \hat{f}_0(\hat{y} \mid x) - f_0(\hat{y} \mid x)\hat{f}_t(\hat{y} \mid x)}{f_0(\hat{y} \mid x) \hat{f}_0(\hat{y} \mid x)} \\
&\leq \frac{(\hat{f}_t(\hat{y} \mid x) + \epsilon )(f_0(\hat{y} \mid x) + \epsilon) - f_0(\hat{y} \mid x)\hat{f}_t(\hat{y} \mid x)}{f_0(\hat{y} \mid x) \hat{f}_0(\hat{y} \mid x)} \\
&\leq \frac{2\epsilon(U + \epsilon)}{L^2/2} \leq \frac{8U\epsilon}{L^2}.
\end{align*}
The first inequality comes from the definition of $y$, the second inequality comes from Assumptions (ii) and (iii), and the last two inequalities follow from Assumption (i). A symmetric line of reasoning that demonstrates $\pi^\star(x) - \hat{\pi}^\star(x) \leq \frac{8U\epsilon}{L^2}$ completes the proof.
\end{proof}

Our second lemma shows that the posterior estimate in \cref{eqn:generalized_w_estimate} is close to the ground truth value.

\begin{lemma}
\label{lem:posterior-approximation}
Let $\epsilon, L, U > 0$ be given. Suppose the following holds for all $i=1, \ldots n$ and $y \in \Ycal$:
\begin{itemize}
    \item[(i)] $L \leq f_t(y \mid x_i), f_0(y \mid x_i) \leq U$,
    \item [(ii)] $|f_t(y \mid x_i) - \hat{f}_t(y \mid x_i)| \leq \epsilon$, and
    \item[(iii)] $|f_0(y \mid x_i) - \hat{f}_0(y \mid x_i)| \leq \epsilon$.
\end{itemize}
If $\epsilon \leq \min (L/2, 1)$, then for all $i=1,\ldots, n$,
\[ \left|\frac{(1 - \pi^\star(x_i)) f_0(y_i \mid x_i)}{f_t(y_i \mid x_i)} - \frac{(1 - \hat{\pi}^\star(x_i)) \hat{f}_0(y_i \mid x_i)}{\hat{f}_t(y_i \mid x_i)} \right| \leq \frac{8U}{L^2} \left( 1 + \frac{12U}{L^2}\right) \epsilon.  \]
\end{lemma}
\begin{proof}
Let $x=x_i$, $y = y_i$. Then we have
\begin{align*}
&\frac{(1 - \pi^\star(x)) f_0(y \mid x)}{f_t(y \mid x)} - \frac{(1 - \hat{\pi}^\star(x)) \hat{f}_0(y \mid x)}{\hat{f}_t(y \mid x)}\\
&\hspace{3em}= \frac{(1 - \pi^\star(x)) f_0(y \mid x) \hat{f}_t(y \mid x) - (1 - \hat{\pi}^\star(x)) \hat{f}_0(y \mid x)f_t(y \mid x)}{f_t(y \mid x)\hat{f}_t(y \mid x)} \\
&\hspace{3em}\leq \frac{ ((1 - \hat{\pi}^\star(x)) + \frac{8U\epsilon}{L^2}) (\hat{f}_0(y \mid x) + \epsilon)(f_t(y \mid x) + \epsilon)- (1 - \hat{\pi}^\star(x)) \hat{f}_0(y \mid x)f_t(y \mid x)}{f_t(y \mid x)\hat{f}_t(y \mid x)} \\
&\hspace{3em}\leq \frac{(U + \epsilon)\epsilon + U\epsilon + \epsilon^2 + \frac{8U^2 \epsilon}{L^2}(U+\epsilon) + \frac{8U \epsilon}{L^2}\epsilon^2 (U + \epsilon) + \frac{8U^2 \epsilon^2}{L^2} + \frac{8U \epsilon^3}{L^2} }{L^2/2} \\
&\hspace{3em}\leq \frac{8U}{L^2} \left( 1 + \frac{12U}{L^2}\right) \epsilon. 
\end{align*}
Here, we have used \cref{lem:w_star-estimate} in the first inequality.
\end{proof}

We now turn to the proof of \cref{thm:general-fdr}, which is nearly identical to the proof of \cref{thm:additive-fdr}.

\begin{proof}[Proof of \cref{thm:general-fdr}]
Let $(x_1, y_1, t_1, h_1), \ldots, (x_n, y_n, t_n, h_n)$ denote draws from \cref{eqn:two_groups} (where $h_1,\ldots, h_n$ are unobserved). From the structure of \cref{eqn:two_groups}, we have that, conditioned on $(x_1, y_1, t_1), \ldots, (x_n, y_n, t_n)$, the random variables $\ind[h_i = 0]$ are Bernoulli random variables with biases $p_i = \pr(h_i = 0 \mid x_i, y_i, t_i)$.

From \Cref{lem:posterior-approximation}, we have that
\[\hat{w}_i \leq \pr(H=0 \mid x_i, y_i, t_i=1) + \frac{8U}{L^2} \left( 1 + \frac{12U}{L^2}\right) \epsilon \]
for all $i=1,\ldots, n$. If $S$ is a set of indices satisfying $t_i=1$ for all $i \in S$ and $\frac{1}{|S|} \sum_{i\in S} \hat{w}_i \leq \alpha$, then
\begin{align*}
\E \left[ \frac{1}{|S|} \sum_{i \in S} \ind[h_i = 0] \mid (x_1, y_1, t_1), \ldots, (x_n y_n, t_n) \right] 
&= \frac{1}{|S|} \sum_{i \in S} \E\left[ \ind[h_i = 0] \mid (x_1, y_1, t_1), \ldots, (x_n y_n, t_n) \right] \\
&=  \frac{1}{|S|} \sum_{i \in 
S} \hat{w}_i + \frac{8U}{L^2} \left( 1 + \frac{12U}{L^2}\right) \epsilon \\
&\leq \alpha + \frac{8U}{L^2} \left( 1 + \frac{12U}{L^2}\right)\epsilon.
\end{align*}
To prove the power statement, we first observe that the Bayes' optimal procedure is to select the largest set such that $\frac{1}{|V|} \sum_{i\in V} p_i \leq \alpha$, where $p_i = \pr(h_i = 0 \mid x_i, y_i, t_i)$ is formed using the prior $\pi^\star$ from \cref{lem:conservative-pi}. Then \cref{lem:prob-subset-sum} finishes the argument.
\end{proof}

\section{Additional results and details from \Cref{sec:confounding}}
\subsection{Proof of \cref{prop:compliance-confounding-ac2g}}

The full compliance confounding additive causal two-groups model is as follows.

\begin{align*}
T |  X=x, U=u &\sim& \textnormal{Bern}(\phi(x,u))  \\
H | X=x, U=u, T=0 &=& 0  \\
H | X=x, U=u, T=1 &\sim& \textnormal{Bern}(\pi(x,u))  \\
Y | X=x, H=h &=& \mu_{h}(x) + \epsilon  \\
\epsilon &\sim& g.
\end{align*}

To show that the additive causal two-groups model is well-specified, observe that we still have
\[ p(Y = y \mid X=x, T=0 ) 
= p(Y = y \mid X=x, H=0 ) 
=  {g}(y - {\mu}_0(x)). \]
Thus, all that remains to be shown is that when we marginalize over $U$, we can write out
\[ p(Y = y \mid X=x, T=1) =  (1-\pi(x)) {g}(y - {\mu}_0(x)) + \pi(x) {g}(y - \mu_1(x)) .\]
Working it out,
\begin{align*}
p(Y = y \mid X=x, T=1) &= \sum_{h =0}^1 p(H=h \mid X=x, T=1) p(Y = y \mid X=x, H=h) \\
&= \sum_{h =0}^1 p(H=h \mid X=x, T=1) g(y - \mu_h(x)).
\end{align*}
Marginalizing over $U$, we have
\begin{align*}
p(H=1 \mid X=x, T=1) &= \int p(U = u \mid X=x, T=1) p(H = 1 \mid X=x, U=u, T=1) \, du \\
&= \int p(U = u \mid X=x, T=1) \pi(x,u) \, du \\
&=: \pi(x),
\end{align*}
where we have simply defined $\pi(x)$ in the last line. Putting it all together gives us the desired result.

\subsection{Proof of \cref{prop:latent-nonrobust-ac2g}}

Consider the following model.
\begin{align*}
H \mid X=x, T=1 &\sim& \textnormal{Bern}(\pi(x)) \\ 
Y \mid X=x, U=u, H=0 &\sim& \mu_0(x) + u + \epsilon \\ 
Y \mid X=x, U=u, H=1 &\sim& \mu_1(x) + \epsilon \\ 
U, \epsilon &\sim& \Ncal(0,1).
\end{align*}
Marginalizing over $U$ and $\epsilon$, we have
\begin{align*}
Y \mid X=x, H=0 &\sim& \Ncal(\mu_0(x),2) \\ 
Y \mid X=x, H=1 &\sim& \Ncal(\mu_1(x),1).
\end{align*}
As the above have different standard deviations, this model cannot be rewritten in the form of \cref{eqn:additive_errors}. Thus, the additive causal two-groups model is misspecified here.
\subsection{Proof of \cref{prop:confounding-npc2g}}

We need to show that $p(Y \mid X, T=1)$ can be written as a mixture model as in \cref{eqn:treatment-mixture}. This can be done by making use of the condition $Y \indep T \mid X, H$:
\begin{align*}
p(Y = y \mid X=x, T=1) &= \sum_{h=0}^1 p(H=h \mid X=x, T=1) p(Y = y \mid X=x, T=1, H=h) \\
&= \sum_{h=0}^1 p(H=h \mid X=x, T=1) p(Y = y \mid X=x, H=h) \\
&= (1-\pi(x)) f_0(y \mid x) + \pi(x) f_1(y \mid x).
\end{align*}
Here, we have simply defined $\pi(x) = p(H=1 \mid X=x, T=1)$, $f_0(y \mid x) = p(Y = y \mid X=x, H=0)$, and $f_1(y \mid x) = p(Y = y \mid X=x, H=1)$. \qed

\subsection{Proof of \cref{prop:canonical-confounding}}

We will make use of the following lemma.

\begin{lemma}
\label{lem:general-latent-confounding}
Let $\Dcal$ and $\Dcal'$ be two distributions. For random variables $X, Y, H, T, U$, there is a choice of distributions for $U$, $Y|X,H,U$, $T|X,U$, and $H|X,T$ such that
\begin{enumerate}
\item $X, Y, H, T, U$ obey the canonical confounding graphical model,
\item $\pr(T=1 \mid X) = 1/2$ [no treatment overlap violations],
\item $\pr(H=0 \mid X, T=0) = 1$ [all untreated are non-responders], 
\item $\pr(H=1 \mid X, T=1) = 1/2$ [no responder overlap violations], 
\item $Y \mid X, T=0 \sim \Dcal'$, and
\item $Y \mid X, H=0, T=1 \overset{d}{=} Y \mid X, H=1, T=1 \overset{d}{=} Y \mid X, T=1 \sim \Dcal$ [matching responder and non-responder distributions].
\end{enumerate}
\end{lemma}
\begin{proof}
Our conditional distributions will be independent of $X$, so we will drop $X$ going forward. We make the following choices for the conditional distributions
\begin{align*}
Y \mid  H=0, U=0 &\sim \Dcal \\
Y \mid H=0, U=1 &= \Dcal' \\
Y \mid H=1, U=0 &= \Dcal \\
Y=1 \mid H=1, U=1 &= \Dcal' \\
\pr(T=1 \mid U=0) &= 1 \\ 
\pr(T = 1 \mid U=1) &= 0 \\ 
\pr(H=1 \mid T=1) &= 1/2 \\
\pr(H=0 \mid T=0) &= 1 \\
\pr(U=1) &= 1/2.
\end{align*}
Then we can work out the following
\begin{align*}
\pr(T=1) &= \pr(U=1)\pr(T = 1 \mid U=1) + \pr(U=0)\pr(T = 1 \mid U=0) = 1/2 \\
\pr(U=1 \mid H, T=1) &= \pr(U=1 \mid T=1) \\
&= \frac{\pr(T=1 \mid U=1)\pr(U=1)}{\pr(T=1 \mid U=1)\pr(U=1) + \pr(T=1 \mid U=0)\pr(U=0)} = 0 \\
\pr(U=1 \mid H, T=0) &= \pr(U=1 \mid T=0) \\
&= \frac{\pr(T=0 \mid U=1)\pr(U=1)}{\pr(T=0 \mid U=1)\pr(U=1) + \pr(T=0 \mid U=0)\pr(U=0)} = 1 \\
\pr(U=0 \mid H, T=1) &= 1 \\
\pr(U=0 \mid H, T=0) &= 0.
\end{align*}
And so we have
\begin{align*}
\pr(Y=y | H=1, T=1) &= \pr(Y=y \mid H=1, U=0, T=1) \pr(U=0 \mid H=1, T=1) \\
&\hspace{2em} +  \pr(Y=y \mid H=1, U=1, T=1) \pr(U=1 \mid H=1, T=1) \\
&= \pr(Y=y \mid H=1, U=0) \pr(U=0 \mid H=1, T=1) \\
&\hspace{2em} +  \pr(Y=y \mid H=1, U=1) \pr(U=1 \mid H=1, T=1) \\
&= \pr(Y=y \mid H=1, U=0).
\end{align*}
Thus, $Y \mid H=1, T=1 \sim \Dcal$. Similarly, we have
\begin{align*}
\pr(Y=y | H=0, T=1) &= \pr(Y=y \mid H=0, U=0, T=1) \pr(U=0 \mid H=0, T=1) \\
&\hspace{2em} +  \pr(Y=y \mid H=0, U=1, T=1) \pr(U=1 \mid H=0, T=1) \\
&= \pr(Y=y \mid H=0, U=0) \pr(U=0 \mid H=0, T=1) \\
&\hspace{2em} +  \pr(Y=y \mid H=0, U=1) \pr(U=1 \mid H=0, T=1) \\
&= \pr(Y=y \mid H=0, U=0).
\end{align*}
Therefore, $Y \mid H=0, T=1 \sim \Dcal$. Moreover, we can work out
\begin{align*}
\pr(Y = y \mid T=0) &= \pr(Y = y \mid H=0, T=0) \\
&= \pr(Y = y \mid U=0, H=0) \pr(U=0 \mid T=0) \\
&\hspace{3em}+ \pr(Y = y \mid U=1, H=0) \pr(U=1 \mid T=0)\\
&= \pr(Y = y \mid U=1, H=0).
\end{align*}
Thus, $Y \mid T=0 \sim \Dcal'$.
\end{proof}

\begin{proof}[Proof of \cref{prop:canonical-confounding}]
To finish the proof of \cref{prop:canonical-confounding}, we can take $\Dcal$ and $\Dcal'$ to be distributions with non-intersecting supports, e.g. uniform distributions over separated intervals. Then \cref{lem:general-latent-confounding} gives us a confounded setting in which the non-treatment distribution does not match the non-response-under-treatment distribution, violating \cref{eqn:two_groups}. Thus, the nonparametric causal two-groups model is not robust against latent confounding, and therefore is also not robust against total confounding.

Moreover, we can also see that
if we were to apply the nonparametric causal two-group model anyway, we would conclude that $\pi(x) = 1$ for all $x$ satisfying $t=1$. This is because the treatment and non-treatment distributions do not share any support, and so the only valid mixture satisfying \cref{eqn:treatment-mixture} is the trivial mixture that assigns the responder distribution to be the treatment distribution itself and for $H=1$ to occur with probability 1 when $T=1$. However, as shown in \cref{prop:canonical-confounding}, the true probability of $H=1$ is 1/2 when $T=1$. 
\end{proof}

\section{Additional results and details from \Cref{sec:simulations}}
\label{sec:more_simulations}
\begin{table}[htp!]
\centering
\begin{tabular}{|| l | l||} 
 \hline
 Additive & Nonadditive \\ 
 \hline\hline
 $\beta, \gamma, \theta \sim \mathcal{N}(0, I_d/\sqrt{d} )$ & $\beta, \gamma \sim \mathcal{N}(0, I_d/\sqrt{d} )$ \\ 
  & $c \sim \text{HalfNormal}(0, 2)$ \\
 $H | T=1 \sim \text{Bern}( \text{sigmoid}(\beta^T X))$ & $W = \gamma^TX$  \\
 $T \sim \text{Bern}(0.5)$ & $T \sim \text{Bern}( \text{sigmoid}(W)$ \\
 $\mu_0 = \gamma^T X $ & $H | T=1 \sim \text{Bern}( \text{sigmoid}(\beta^T X))$ \\ 
 $\mu_1 = \mu_0 + \tau \sum_i |X_i||\gamma_i| $ & $B_{ij} \sim \text{Bern}(0.1), Z_{ij} \sim \text{StudentT}(3)$ \\
 $Y | H=0 \sim \mathcal{N}(\mu_0, 1)$  & $L = \sum_{i,j} B_{i,j} Z_{i,j} X_i X_j$ \\
 $Y | H=1 \sim \mathcal{N}(\mu_1, 1)$   & $Y \sim \mathcal{N}( \log(1 + \exp(cW + \theta^T X + \tau H + L)), 1 ) $ \\
 \hline
\end{tabular}
\caption{Data generating models for additive and nonadditive synthetic simulations.}
\label{table:synthetic-data}
\end{table}

The full data generating process for both synthetic additive and nonadditive simulations is described in \cref{table:synthetic-data}.
All simulations were performed on a cluster with Intel Xeon Gold 6348 CPUs. Each simulation ran in under 2 hours with 2 CPUs allocated. In all plots, shaded regions represent 95\% confidence intervals under the normal approximation.

For additional baselines, we also compared against three methods that are valid ways of testing in the intention-to-treat setting. FDR-Regression \citep{scott:etal:2014:fdr-regression} controls the local FDR in different subsets of the experiment. BART \citep{hill:etal:2011:causal-bart} and Causal forests \citep{wager:athey:2018:causal-forests} are two nonparametric tree-based methods for estimating individual treatment effects. 

To appropriately measure power in this setting, we considered \emph{valid power}, which we define to be 
\[ \text{vp}(\alpha) =
\begin{cases}
0 & \text{ if }  \, \, \hat{\text{fdr}}(\alpha) - \text{CI}_{\text{fdr}} > \alpha \\
\frac{\# \text{ of rejected true non-nulls}}{\# \text{of true non-nulls}} & \text{ otherwise}
\end{cases},  \]
where $\hat{\text{fdr}}(\alpha)$ is the average observed FDR of the procedure at level $\alpha$ and $\text{CI}_{\text{fdr}}$ is its associated 95\% confidence interval band, over the 50 different random runs. Intuitively, valid power captures the standard notion of power in settings where FDR has been respected.

We generally observe pathological behavior for the intention-to-treat baselines. In particular, both BART and Causal forest generally reject everything at extremely low nominal FDR levels, leading to high observed FDR rates. The valid power for these methods is categorically zero until the first nominal FDR level for which the observed FDR is acceptable, at which point the valid power immediately jumps to 1. On the other hand, we observe that FDRreg's behavior changes between settings, having high FDR values on the additive data and low FDR values on the nonadditive data. In all settings, however, the valid power of FDRreg was far below the C2G methods and the frequentist baseline.

It is worth emphasizing that these results are not intended as an indictment of the baseline methods. When the assumptions of the problem match those of the methods, these baselines generally perform quite well. Rather, these results illustrate how the causal two-groups setting breaks the assumptions behind these methods.

\begin{figure*}
\centering %
\hspace{-2em}\includegraphics[width=1.03\textwidth]{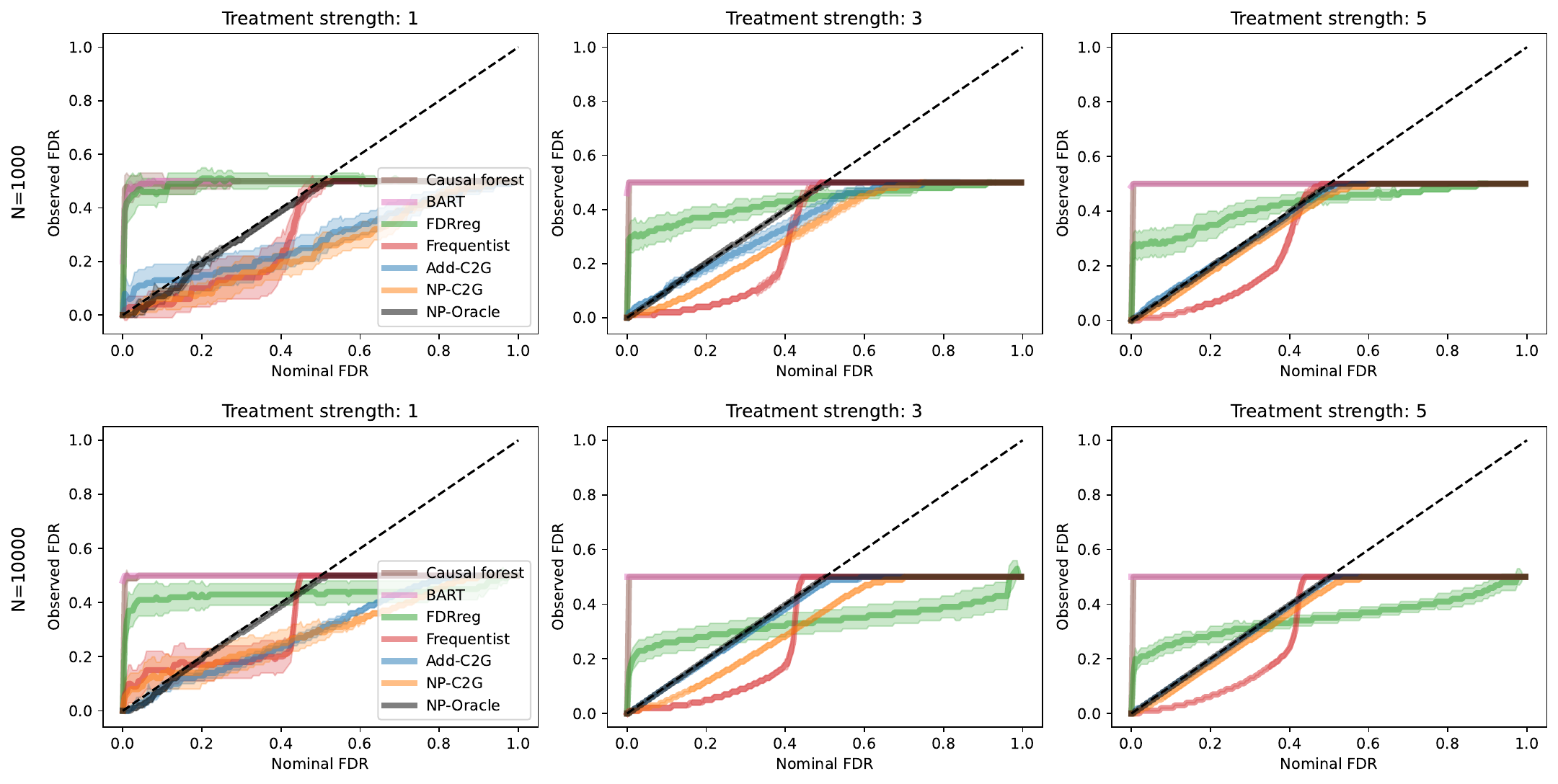}
\caption{FDR curves on additive synthetic data.}
\end{figure*}

\begin{figure*}
\centering %
\hspace{-2em}\includegraphics[width=1.03\textwidth]{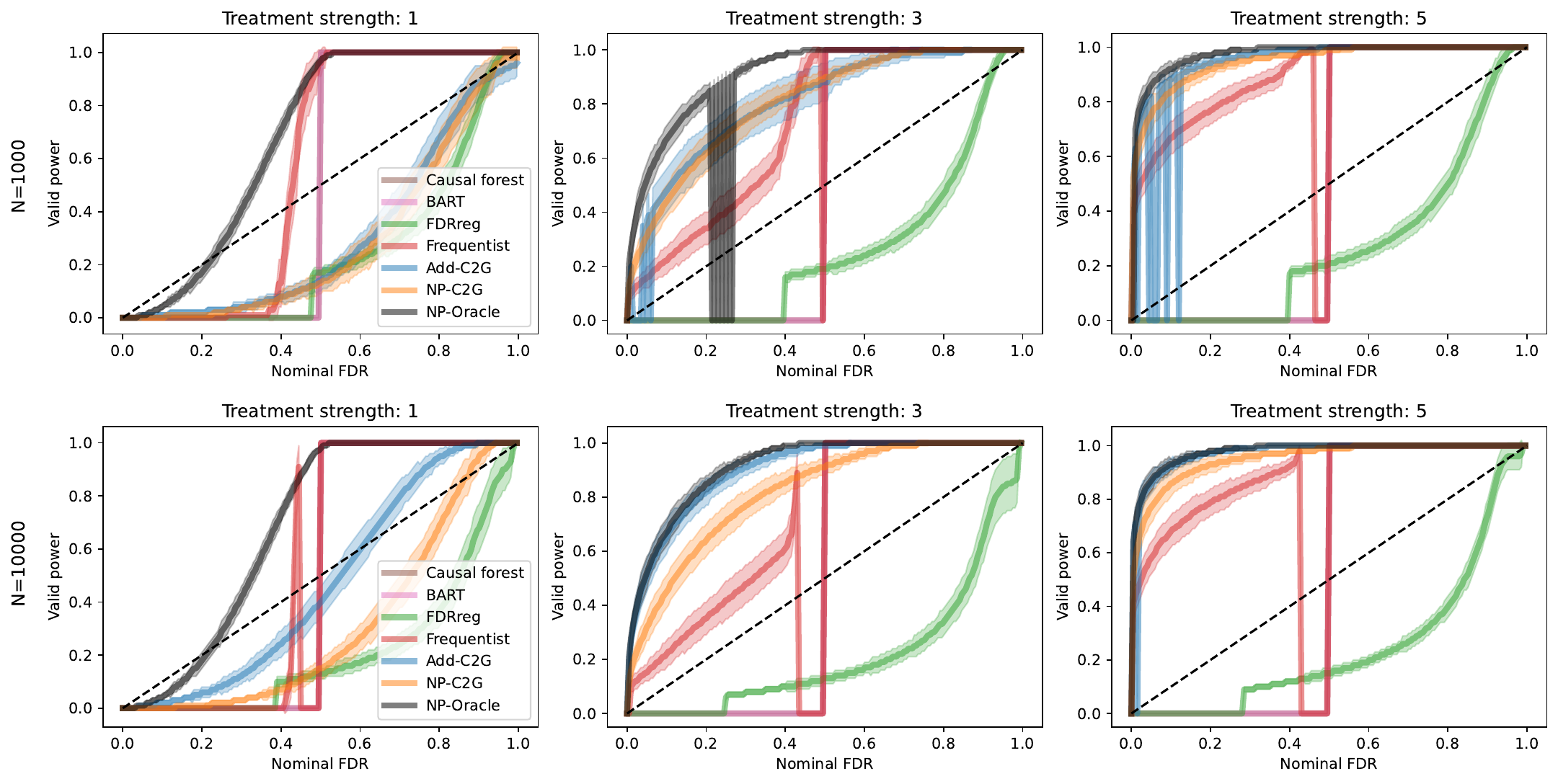}
\caption{Valid power curves on additive synthetic data.}
\end{figure*}

\begin{figure*}
\centering 
\hspace{-2em}\includegraphics[width=1.03\textwidth]{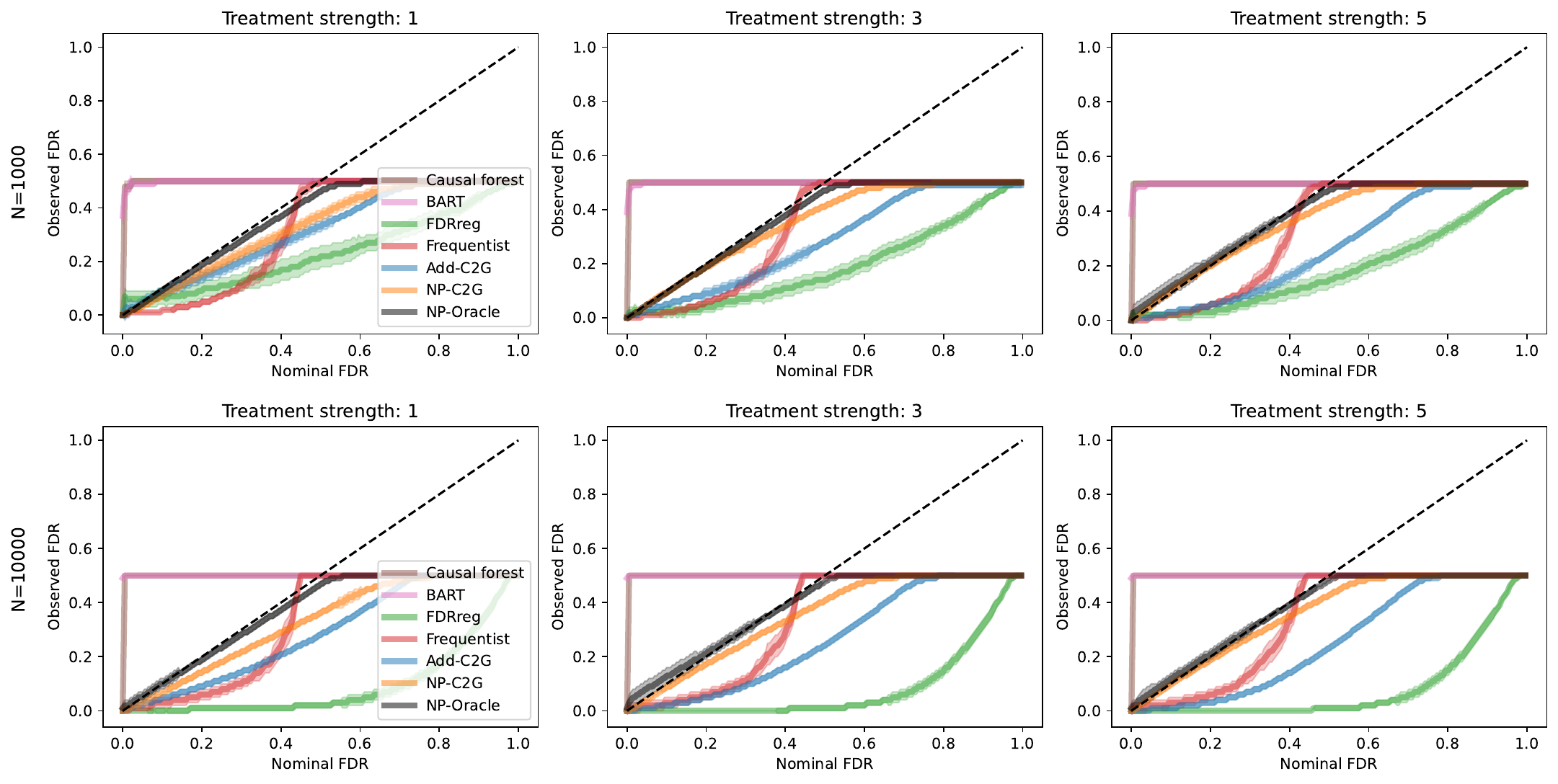}
\caption{FDR curves on nonadditive synthetic data.}
\end{figure*}

\begin{figure*}
\centering %
\hspace{-2em}\includegraphics[width=1.03\textwidth]{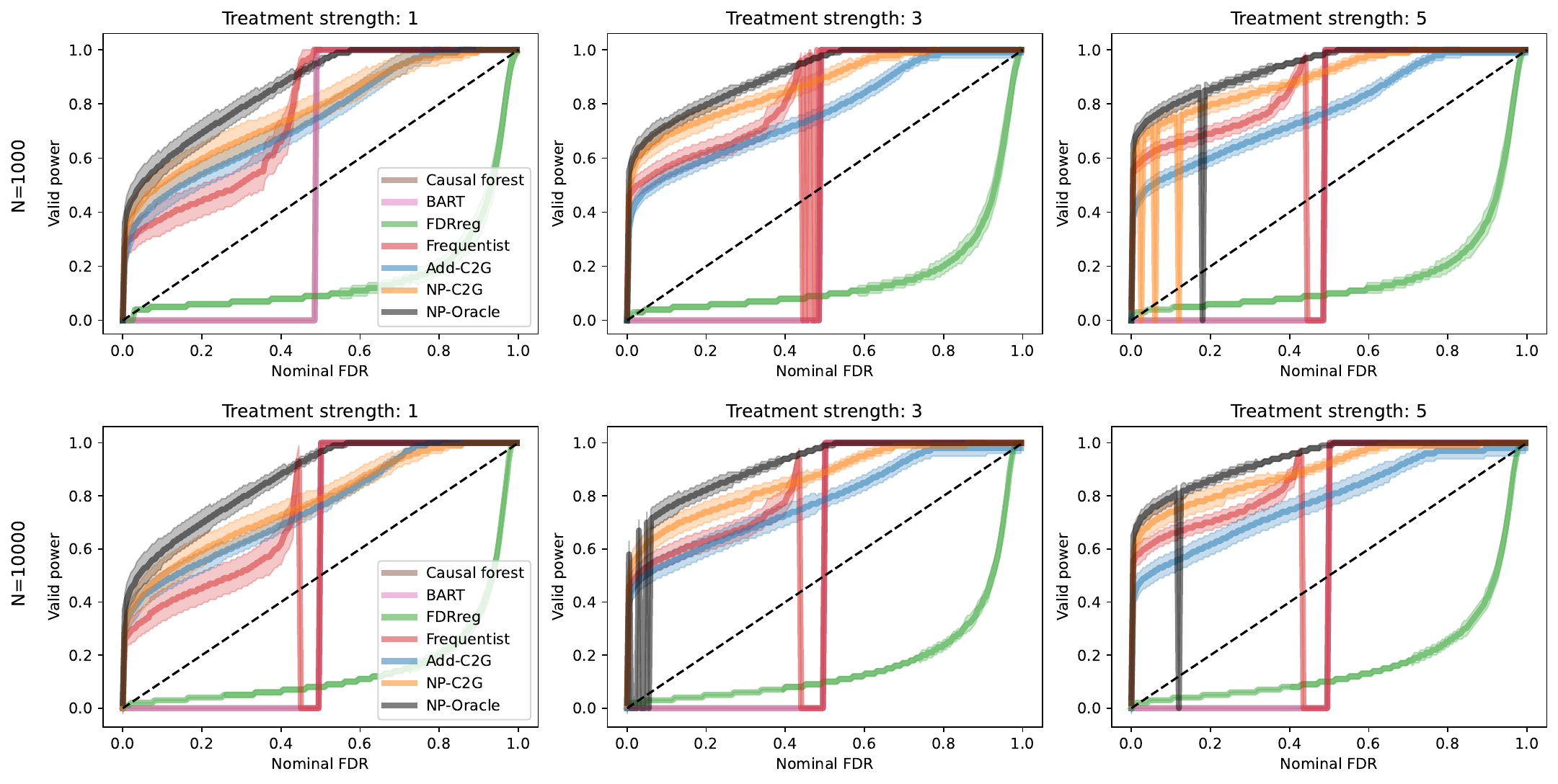}
\caption{Valid power curves on nonadditive synthetic data.}
\end{figure*}

\section{Semi-synthetic simulations on drug sensitivity data}
\label{sec:gdsc_simulations}

\begin{table}
\centering
\caption{Empirical results on semi-synthetic drug-response data. $\pm$ denotes 95\% confindence intervals. \textbf{Bolded} results on FDR column indicate that the method(s) achieved valid FDR (up to 95\% confidence intervals). \textbf{Bolded} results on valid power column indicate that the method(s) achieved the highest valid power for that setting (up to 95\% confidence intervals).}
\begin{tabular}{||c|c|c|c|c|c|c||}
\hline
& \multicolumn{3}{|c|}{FDR} & \multicolumn{3}{|c|}{Valid power} \\
\hline
Method & $\alpha=0.05$ & $\alpha=0.1$ & $\alpha=0.25$ & $\alpha=0.05$ & $\alpha=0.1$ & $\alpha=0.25$ \\
\hline
Frequentist & \textbf{ 0.03±0.006 } & \textbf{ 0.04±0.007 } & \textbf{ 0.19±0.038 } & 0.24±0.018 & 0.33±0.021 & \textbf{ 0.63±0.062 } \\
Add-C2G & 0.14±0.035 & 0.17±0.034 & \textbf{ 0.23±0.032 } & 0.0±0.0 & 0.0±0.0 & 0.42±0.06 \\
NP-C2G & \textbf{ 0.04±0.009 } & \textbf{ 0.07±0.011 } & \textbf{ 0.23±0.02 } & \textbf{ 0.33±0.025 } & \textbf{ 0.44±0.023 } & \textbf{ 0.65±0.021 } \\
\hline
\end{tabular}
\label{table:nutlin-simulations}
\end{table}
The Genomics of Drug Sensitivity in Cancer (GDSC) is a repository of cancer cell line data, both genomic and drug-response~\cite{iorio:etal:2016:gdsc-interactions,yang:etal:2013:gdsc}. We considered a subset of 832 cell lines for which both genomic data and dose-response data against the drug Nutlin-3a were recorded. Nutlin-3a is an MDM2-inhibitor, and as such it promotes p53, a tumor suppressing protein~\cite{shen:maki:2011:p53}. However, this pathway can successfully inhibit cancer growth only when the corresponding gene, TP53, is wild-type, i.e. does not contain a missense mutation.

For each cell line in our dataset, we took the covariates $X$ to be the 50-dimensional PCA projection of its RNA expression data. We took the outcome variable $Y$ to be the z-scored cell line viability of the cell line when exposed to the maximum dose of Nutlin-3a. We let $H$ be an indicator variable that is 1 when the cell line does not have a TP53 mutation. And we took $T$ to be 1 whenever $H=1$, and otherwise we set it to a Bernoulli draw with bias equal to $\text{sigmoid}(W_{\text{GAPDH}})$, where $W_{\text{GAPDH}}$ is the RNA expression of the Glyceraldehyde 3-phosphate dehydrogenase gene. We reran our simulations with 50 different random seeds.

\cref{table:nutlin-simulations} displays the empirical results on GDSC data. Note that in this setting, all methods did a reasonable job of controlling FDR, save for Add-C2G at the lower levels of $\alpha$. However, for BART, FDRreg, and Causal Forests, this is mostly due to having an unreasonably low rejection rate. The fact that Add-C2G generally does not control FDR in this setting likely indicates that the additivity assumption is grossly violated here. Finally, we see that NP-C2G controlled FDR at all levels while maintaing high valid power.

\begin{figure*}
\centering 
\includegraphics[width=1.0\textwidth]{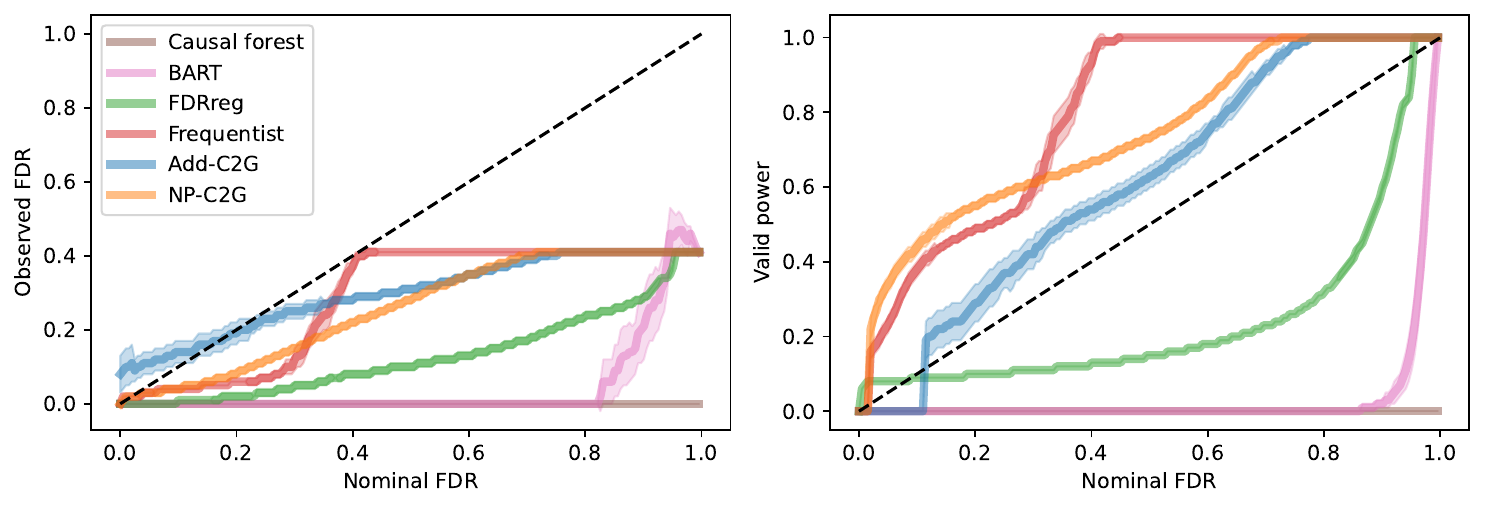}
\caption{Observed FDR (left) and valid power curves (right) for semi-synthetic drug-response data.}
\end{figure*}

\section{Additional results and details from \Cref{sec:case-study}}
\cref{fig:q-val} shows the q-values of individual genes as a function of the nominal FDR level used in NP-C2G selection.

\begin{figure*}
\centering
\includegraphics[width=1.0\textwidth]{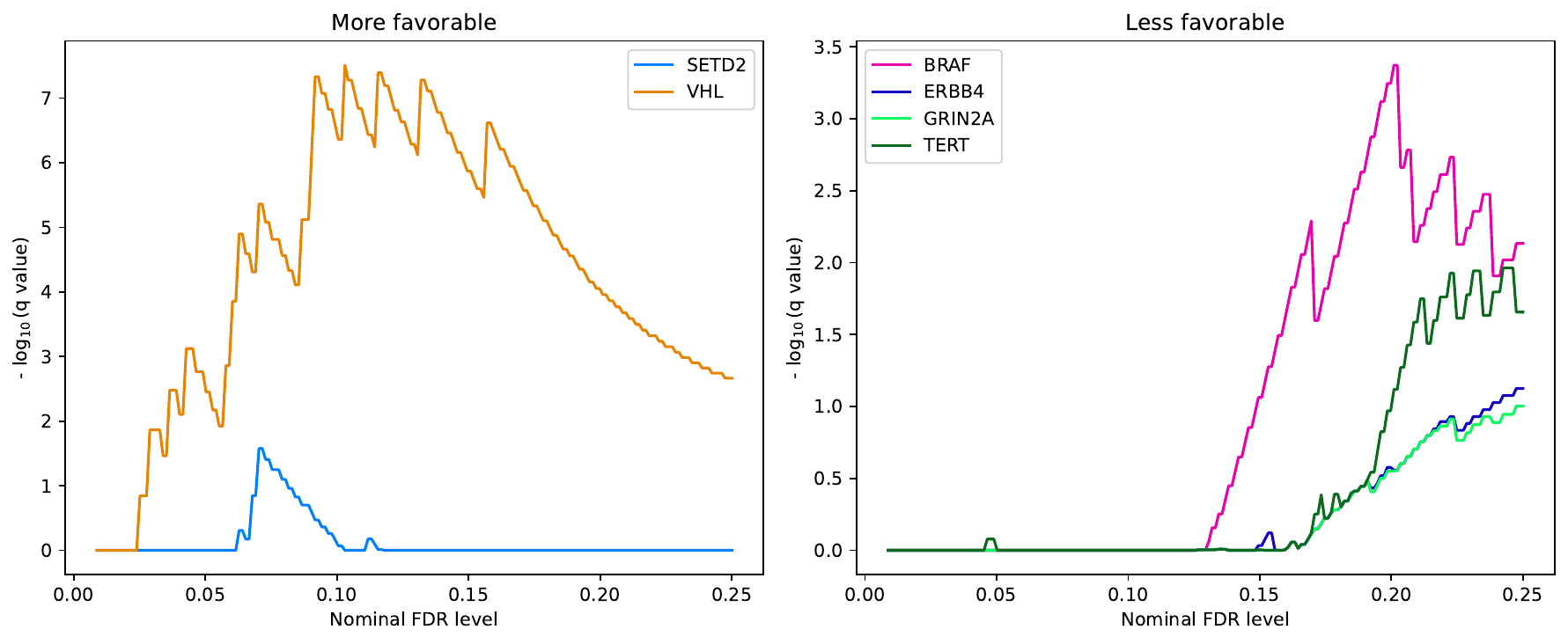}
\caption{q-values as a function of nominal FDR level for each gene. Only genes achieving q-value below 0.1 are shown.}
\label{fig:q-val}
\end{figure*}

\end{document}